%% file: jbes2025.tex
\documentclass[12pt]{article}
\usepackage{amsmath, amssymb, amsfonts, amsthm}
\usepackage{graphicx}
\usepackage{subfigure}
\usepackage{enumerate}
\usepackage{url} 
\usepackage{bbm, bm}
\usepackage{times, enumitem}
\usepackage{mathtools}
\usepackage{upgreek}
\usepackage{caption}
\usepackage{hyperref}
\usepackage{xcolor}
\usepackage[normalem]{ulem} 
\usepackage{setspace}

\usepackage[round]{natbib}
\bibpunct{(}{)}{;}{a}{}{,}

\newcommand{\blind}{0}

\newtheorem{theorem}{Theorem}[section] 
\newtheorem{lemma}{Lemma}[section] 

\newtheorem{proposition}{Proposition}[section]

\newtheorem{remark}{Remark}

\makeatletter
\newcommand{\customlabel}[2]{%
   \protected@write \@auxout {}{\string \newlabel {#1}{{#2}{\thepage}{#2}{#1}{}} }%
   \hypertarget{#1}{#2}
}
\makeatother

\newcommand{\E}{\mathbb{E}}
\newcommand{\Pb}{\mathbb{P}}

\newcommand{\var}{\text{var}}
\newcommand{\cov}{\text{cov}}

\newcommand{\Pn}{\mathbb{P}_n}

\newcommand{\R}{\mathbb{R}}

\DeclareMathOperator{\Beta}{Beta}
\DeclareMathOperator{\diag}{diag}           
\DeclareMathOperator{\expit}{expit}
\DeclareMathOperator{\Unif}{Unif}
\def\ind{\perp\!\!\!\perp}

\newcommand{\sol}{\mathsf{s}^*}

\newcommand{\cmmnt}[1]{\ignorespaces}  
\newcommand{\mishler}[1]{\textcolor{red}{#1}}
\newcommand{\kim}[1]{\textcolor{blue}{#1}}

\begin{document}

\def\spacingset#1{\renewcommand{\baselinestretch}%
{#1}\small\normalsize} \spacingset{1}


\if0\blind
{
  \title{\bf Counterfactual Mean-variance Optimization}
  \author{Kwangho Kim\\
    Department of Statistics, Korea University\\
    \url{kwanghk@korea.ac.kr}\\
    and \\
    Alan Mishler\\
    J. P. Morgan AI Research\\
    \url{alan@alanmishler.com}\\
    and \\
    José R. Zubizarreta\\
    Departments of Health Care Policy, Biostatistics, Statistics, Harvard University\\
    \url{zubizarreta@hcp.med.harvard.edu}
    }
   \date{}
  \maketitle
} \fi

\if1\blind
{
  \bigskip
  \bigskip
  \bigskip
  \begin{center}
    {\LARGE\bf Counterfactual Mean-variance Optimization}
\end{center}
  \medskip
} \fi

\bigskip
\begin{abstract}
We study a counterfactual mean-variance optimization, where the mean and variance are defined as functionals of counterfactual distributions. The optimization problem defines the optimal resource allocation under various constraints in a hypothetical scenario induced by a specified intervention, which may differ substantially from the observed world. We propose a doubly robust-style estimator for the optimal solution to the counterfactual mean-variance optimization problem and derive a closed-form expression for its asymptotic distribution. Our analysis shows that the proposed estimator attains fast parametric convergence rates while enabling tractable inference, even when incorporating nonparametric methods. We further address the calibration of the counterfactual covariance estimator to enhance the finite-sample performance of the proposed optimal solution estimators. Finally, we evaluate the proposed methods through simulation studies and demonstrate their applicability in real-world problems involving healthcare policy and financial portfolio construction.
\end{abstract}

\noindent%
\vfill

\newpage
\spacingset{1.8} 

\section{Introduction}

\emph{Counterfactuals}, also referred to as \emph{potential outcomes}, represent the hypothetical responses of a unit to a particular treatment or intervention, regardless of whether the intervention is actually administered. While counterfactuals have been the dominant causal language in statistics \citep{rubin1974estimating, holland1986statistics, hofler2005causal}, it has more recently emerged as a valuable tool in predictive modeling, particularly for improving decision-making under shifts in treatment patterns \citep[e.g.,][]{wang2019equal, dickerman2020counterfactual, lin2021scoping, dickerman2022predicting, kim2022doubly, kim2025semiparametric}. In this work, we explore a counterfactual extension of the traditional optimal resource allocation problem, formulating a domain-general counterfactual mean-variance optimization framework that parameterizes the tradeoff between the mean and variance of outcomes under a hypothetical intervention.

\subsection{Related Work}

Traditional, non-counterfactual mean-variance optimization has been used to estimate optimal allocations of resources in various settings, including financial investment \citep{markowitz1968portfolio}, {decision theory \citep{meyer1987two}}, product development \citep{cardozo1983ApplyingFinancialPortfolio}, healthcare policy \citep{fagefors2021ApplicationPortfolioTheory, qu2012MeanVarianceModel}, and electrical engineering \citep{delarue2011ApplyingPortfolioTheory}. In counterfactual mean-variance optimization, our goal is to estimate optimal allocations under hypothetical scenarios that may alter the outcome distribution, such as a healthcare policy intervention that affects the compliance rate for scheduled patient appointments (see Section~\ref{subsec:medical-appt}).

In this work, we analyze the counterfactual mean-variance optimization problem by formulating it as a quadratic program (QP), wherein both the objective function and constraints are defined in terms of functionals over a counterfactual distribution. Optimization problems of this counterfactual nature have been widely studied within the causal inference literature, including applications to policy evaluation \citep[e.g.,][]{kitagawa2018should, athey2021policy}, optimal treatment regimes under resource constraints \citep[e.g.,][]{luedtke2015optimal, luedtke2016optimal}, and algorithmic fairness \citep[e.g.,][]{mishler2021fade, mishler2021fairness, coston2020counterfactual}. Unconstrained variants of such problems also often arise in non-standard causal effect estimation via projection methods, where complex nonparametric estimands are approximated within parametric model classes \citep[e.g.,][]{neugebauer2007nonparametric, semenova2021debiased, kennedy2021semiparametric, mcclean2024nonparametric}. More recently, counterfactual prediction has been studied in a similar spirit, with the goal of minimizing counterfactual loss functions under a set of user-specified constraints \citep{kim2022doubly, kim2025semiparametric}. Nonetheless, and perhaps surprisingly, the formulation and analysis of mean-variance optimization under counterfactual scenarios remain largely unexplored.

A key complication arises from the fact that our estimand is defined as the optimal solution to a specialized form of stochastic program, whose coefficients depend on an unobservable counterfactual distribution. This dependence precludes the direct application of standard stochastic programming techniques such as stochastic approximation and sample average approximation methods \citep[see][Section 2]{kim2025semiparametric}. Notably, accurate estimation of the stochastic components of the optimization problem does not necessarily ensure accurate recovery of the corresponding optimal solution. Moreover, covariance matrix calibration is often necessary, especially in high-dimensional settings. However, it remains unclear whether conventional calibration procedures are valid or effective in our counterfactual setting.

\subsection{Contribution and Outline}
We study counterfactual mean-variance optimization as a new tool for informing decision-making under hypothetical scenarios, particularly those not observed in the present but potentially arising in the future, by leveraging recent advances in the counterfactual prediction literature. In Section~\ref{sec:problem}, we introduce the framework, along with the notations and assumptions. Section~\ref{sec:estimation-inference} presents a flexible, doubly robust nonparametric estimator for the optimal solution to the counterfactual mean-variance program. In Section~\ref{sec:calibration-Sigma}, we address the problem of calibrating the estimated counterfactual covariance matrix to enhance the finite-sample performance of the proposed estimators.  This component of the analysis is of independent interest, given the foundational role that covariance matrix estimation occupies in a broad spectrum of statistical and data science applications. In Section~\ref{sec:numerical-illustraction}, we evaluate the finite-sample properties of our estimators through simulation and apply our method to real-world problems in healthcare and finance. In the healthcare application, we examine the optimal proportion of same-day appointments under varying reminder systems that influence patient no-show rates differently. In the financial application, we study optimal portfolio allocation under different federal interest rate regimes. Section~\ref{sec:discussion} concludes.

\section{Problem Formulation} \label{sec:problem}

Suppose that we observe an i.i.d. sample $(Z_{1}, \ldots, Z_{n})$ of tuples $Z = (Y, A, X) \sim \mathbb{P}$, where $Y = (Y_1, \ldots, Y_k) \in \mathbb{R}^k$ denotes multiple outcomes for some fixed and finite $k$, $X \in \mathcal{X} \subset \mathbb{R}^d$ represents a vector of covariates, and $A \in \mathcal{A} = {0, 1}$ is a binary treatment indicator. For simplicity, we assume a binary treatment setting; however, the proposed framework is, in principle, extensible to multi-valued treatments. For $1 \leq i \leq k$, we let $Y^a$ denote the counterfactual that would be observed under treatment $A = a$, $a \in \mathcal{A}$. We focus on the following \emph{counterfactual mean-variance optimization} problem, in which the goal is to balance the “reward” (i.e., the mean) against the “risk” (i.e., the variance) under a counterfactual scenario where the treatment variable $A$ is set to a fixed value $a$:
\begin{equation}
\label{eqn:couterfactual-MV}
\begin{aligned}
    & \underset{w \in \R^k}{\text{minimize}} \quad {1}/{2}w^\top\Sigma^a w  - \lambda w^\top m^a\\
    & \text{subject to } \quad w \in \mathcal{S}^a, \\
\end{aligned}     \tag{$\mathsf{P}_{\text{MV}}$}  
\end{equation}
for $m_a = (m^a_1, ... , m^a_k)^\top$ and $\Sigma^a = \left( \Sigma^a_{ij} \right)$, where
$m^a_i = \E[Y^a_i]
$, $\Sigma^a_{ij} = \text{cov}(Y^a_i,Y^a_j)$, $1 \leq i,j \leq k$. Note that the counterfactual parameters $m^a$ and $\Sigma^a$ are unknown and must be estimated.
The set of constraints $\mathcal{S}^a$ is, by default, defined as
\begin{align} \label{eqn:default-form-constraints}
    \mathcal{S}^a = \{w \mid w^\top 1 = 1, \,  w \geq 0, \, w^\top m^a \geq \text{r}_{min} \},
\end{align}
which ensures that the solution is a vector of convex weights yielding a counterfactual mean no less than a user-specified threshold $\text{r}_{min} \in \R$. While feasible sets of this form are commonly used in conventional mean-variance optimization problems, our framework can accommodate a broader class of constraints. For example, one may consider a general set of linear constraints
\begin{align} \label{eqn:general-form-constraints}
\mathcal{S}^a = \left\{ w \in \mathbb{R}^k \;\middle\vert\; B^a w \leq c^a \right\},
\end{align}
where each element of the matrix $B^a \in \R^{r \times k}$ and the vector $c^a \in \R^r$ can be defined as a functional of the counterfactual distribution under treatment level $A = a$, without affecting the validity of the subsequent results. Since \eqref{eqn:default-form-constraints} is a special case of \eqref{eqn:general-form-constraints}, we refer to \eqref{eqn:general-form-constraints} as our constraint set throughout the paper.

This will be discussed in greater detail in the following section. $\lambda \geq 0$ is a user-determined risk tolerance coefficient that quantifies {their} tolerance towards the risk. The higher the value of $\lambda$, the larger the variance the user is willing to tolerate in order to maximize the reward.

To identify $m^a$ and $\Sigma^a$, that is, to express these counterfactual quantities in terms of the observed data distribution $\Pb$, we impose the following standard causal assumptions \citep[e.g.,][Chapter 12]{imbens2015causal}: for any $a \in \mathcal{A}$,
\begin{itemize}[label={}, leftmargin=*]
	\item \customlabel{assumption:c1}{(C1)} \emph{Consistency}: $Y=Y^a$ if $A=a$
	\item \customlabel{assumption:c2}{(C2)} \emph{No unmeasured confounding}: $A \ind Y^a \mid X$
	\item \customlabel{assumption:c3}{(C3)} \emph{Positivity}: $\Pb(A=a|X)>0  \text{ a.s. }$
\end{itemize}
Assumptions \ref{assumption:c1} - \ref{assumption:c3} are maintained throughout this paper. Under these assumptions, for all $i,j \in \{1,...,k\}$, $\E[Y^a_i]$ and $\E[Y^a_iY^a_j]$ are identified as $\E\{\E[Y_i \mid X, A=a]\}$ and $\E\{\E[Y_iY_j \mid X, A=a]\}$, respectively, and thus estimable from the observed sample. Standard estimation strategies in causal inference can be employed to estimate the counterfactual components $m^a$ and $\Sigma^a$, as discussed in detail in the following section.

\textbf{Notation.}
\cmmnt{Here we give some notation used throughout this paper. When used as subscript letters, we let $t$ and $i,j$ generally denote indices for observed sample and different outcomes, respectively.} For clarity, we use the subscripts $i, j$ only to index the different outcome variables $Y_1, \ldots, Y_k$, and we reserve the subscript $t$ to index the samples $Z_1, \ldots Z_n$. For any fixed vector $v$ and matrix $M$, we let $\Vert v \Vert_2$ and $\Vert M \Vert_F$ denote the Euclidean norm (or $L_2$-norm) and Frobenius norm, respectively. $\Vert \cdot \Vert_2$ is understood as the spectral norm when it is used with a matrix. \cmmnt{We let $\mathbb{B}_{\delta}(\bar{z})$ denote the open ball with radius $\delta > 0$ around the point $\bar{z}$ with \kim{\sout{\mishler{metric}}} $\Vert \cdot \Vert_2$ (unless otherwise mentioned), i.e., $\mathbb{B}_{\delta}(\bar{z})= \{z \mid \Vert z - \bar{z} \Vert_2 < \delta \}$.} Let $\Pn$ denote the empirical measure over $(Z_1,...,Z_n)$. Given a fixed operator $h$ (e.g., an estimated function), we let $\Pb$ denote the conditional expectation over a new independent observation $Z$, as in $\Pb(h)=\Pb\{h(Z)\}=\int h(z)d\Pb(z)$. Further, we use $\Vert h \Vert_{2,\Pb}$ to denote the $L_2(\Pb)$ norm of $h$ defined by $\Vert h \Vert_{2,\Pb} = \left[\Pb (h^2) \right]^{1/2} = \left[\int h(z)^2 d\Pb(z)\right]^{1/2}$. Lastly, we let $\sol(\mathsf{P})$ denote the set of optimal solutions of an optimization program $\mathsf{P}$. 

\section{Estimation and Inference} \label{sec:estimation-inference}

Leveraging tools from semiparametric theory in causal inference and recent advances in counterfactual prediction, we develop a nonparametric estimator for the optimal solution to \ref{eqn:couterfactual-MV}, and establish that it attains $\sqrt{n}$ convergence rates and asymptotic normality under mild regularity conditions. To this end, we first present efficient estimators for the counterfactual mean ($m^a$) and covariance ($\Sigma^a$). To simplify notation, we introduce the following nuisance functions
\begin{align*}
    & \pi_a(X)=\Pb[A=a \mid X], \\
    & \mu_i(X,a)=\E[Y_i \mid X, A=a], \\
    & \sigma_{ij}(X,a) = \E[Y_iY_j \mid X, A=a],
\end{align*}
and let $\widehat{\pi}_a, \widehat{\mu}_i,$ and $\widehat{\sigma}_{ij}$ be some estimators of ${\pi}_a, {\mu}_i,$ and ${\sigma}_{ij}$, respectively. Further, we let
\begin{align*}
    & \phi^a_i(Z;\eta_{i}) = \frac{\mathbbm{1}(A=a)}{\pi_a(X)}\left\{Y_i - \mu_i(X,A)\right\} + \mu_i(X,a), \\
    & \phi^a_{ij}(Z;\eta_{ij}) = \frac{\mathbbm{1}(A=a)}{\pi_a(X)}\left\{Y_iY_j - \sigma_{ij}(X,A)\right\} + \sigma_{ij}(X,a),
\end{align*}
denote the uncentered efficient influence functions for the parameters ${\psi}^a_{i} \coloneqq \E[Y^a_i] = \E\{\E[Y_i \mid X, A=a]\}$ and ${\psi}^a_{ij} \coloneqq \E[Y^a_iY^a_j] = \E\{\E[Y_iY_j \mid X, A=a]\}$, with the relevant nuisance functions collectively denoted by $\eta_{i}=\{\pi_a(X), \mu_i(X,A)\}$ and $\eta_{ij}=\{\pi_a(X), \sigma_{ij}(X,A)\}$, respectively. 

Estimation of mean counterfactual outcomes such as ${\psi}^a_{i}$ under the standard identification assumptions \ref{assumption:c1} - \ref{assumption:c3} has been extensively studied in the causal inference literature, e.g., for the average treatment effect. The most commonly used estimators include the plug-in (PI) regression, also known as g-computation, and inverse probability weighting (IPW), whose estimation errors are directly influenced by the convergence rates of the underlying nuisance estimators. Semiparametric (or doubly robust) estimators are another well-established class of methods. They can be viewed as augmented versions of plug-in or IPW estimators, incorporating an additional bias-correction term to improve robustness and efficiency \citep{robins1995semiparametric, robins2000inference}. Semiparametric estimators possess several appealing properties: (1) they can attain the fast parametric $\sqrt{n}$ convergence rate even when all nuisance functions are estimated flexibly at slower, nonparametric rates, and (2) they achieve asymptotic normality with semiparametric efficiency under standard regularity conditions \citep[][]{kennedy2016semiparametric, kennedy2024semiparametric}.

For $m^a_i = \E[Y^a_i]$, ${\Sigma}^a_{ij} = \E[Y^a_iY^a_j] - \E[Y^a_i]\E[Y^a_j]$, the corresponding semiparametric estimators are constructed as
\begin{equation} \label{eqn:return-estimator}
\begin{aligned} 
    \widehat{m}^a_i(Z;\widehat{\eta}_i)  = \Pn\left\{ \widehat{\phi}^a_i(Z) \right\},
\end{aligned}
\end{equation}
\begin{equation} \label{eqn:cov-estimator}
\begin{aligned}
    \widehat{\Sigma}^a_{ij}(Z;\widehat{\eta}_i, \widehat{\eta}_j, \widehat{\eta}_{ij})
    = \Pn\left\{\widehat{\phi}^a_{ij}(Z)\right\}  - \Pn\left\{\widehat{\phi}^a_i(Z) \right\}
    \Pn\left\{\widehat{\phi}^a_j(Z)\right\},
\end{aligned}
\end{equation}
respectively, where $\widehat{\phi}^a_i(Z) = {\phi}^a_i(Z; \widehat{\eta}_i)$ and $\widehat{\phi}^a_{ij}(Z) = {\phi}^a_{ij}(Z; \widehat{\eta}_{ij})$. We now state the following assumptions pertaining to the nuisance function estimators:
\begin{enumerate}[label={}, leftmargin=*]
	    \item \customlabel{assumption:B1}{(B1)} $\widehat{\pi}_a, \widehat{\mu}_i, \widehat{\mu}_j, \widehat{\sigma}_{ij}$ are constructed using a single separate iid sample $\mathsf{D}_0^n = \{Z_{n+1}, ..., Z_{2n}\}$
	    \item \customlabel{assumption:B2}{(B2)}  $\Pb(\widehat{\pi}_a \in [\epsilon, 1-\epsilon]) = 1$ for some $\epsilon > 0$ 
	    \item \customlabel{assumption:B3}{(B3)} $\Vert \widehat{\mu}^a_i - {\mu}^a_i \Vert_{2,\Pb} = o_\Pb(1)$, $\Vert \widehat{\pi}^a_i - {\pi}^a_i \Vert_{2,\Pb} = o_\Pb(1)$, $\Vert \widehat{\sigma}^a_{ij} - {\sigma}^a_{ij} \Vert_{2,\Pb} = o_\Pb(1)$
	    \item \customlabel{assumption:B4}{(B4)}  
	    $   \begin{aligned}[t]
            &\Vert \widehat{\pi}_a - \pi_a \Vert_{2,\Pb} \max_{i} \Vert \widehat{\mu}_i - \mu_i \Vert_{2,\Pb} = o_\Pb(n^{-1/2}), \\
            &\Vert \widehat{\pi}_a - \pi_a \Vert_{2,\Pb} \max_{i,j} \Vert \widehat{\sigma}_{ij} - \sigma_{ij} \Vert_{2,\Pb} = o_\Pb(n^{-1/2})
            \end{aligned} 
        $
\end{enumerate}	    

The above assumptions are commonly used in semiparametric approaches in causal inference.
In this work, we employ sample splitting as described in \ref{assumption:B1} to permit the use of arbitrarily complex nuisance estimators \citep{kennedy2016semiparametric, Chernozhukov17, Chernozhukov18}. Specifically, the nuisance functions are estimated on an independent sample of size $n$, distinct from the estimation sample on which $\mathbb{P}_n$ operates; in fact, it suffices that the auxiliary sample be of order $O(n)$, that is, of the same asymptotic order as the estimation sample. (See Remark \ref{rmk:sample-splitting}). If one is willing to rely on appropriate empirical process conditions (e.g., Donsker or low-entropy type conditions \citep{van2000asymptotic}), then the nuisance estimators can be estimated on the same sample without \ref{assumption:B1}; however, this would limit the flexibility of the nuisance estimators. The requirement \ref{assumption:B4} that the second-order nuisance errors converge to zero at faster than $\sqrt{n}$ rates is a sufficient condition commonly found in standard semiparametric settings with finite-dimensional parameters \citep[e.g.,][]{kennedy2020optimal,kennedy2024semiparametric}. 

\begin{remark}[Sample splitting] \label{rmk:sample-splitting}
For nuisance estimation, we can always create separate independent samples by splitting the data in half (or in folds) at random; furthermore, full sample size efficiency can be attained by
swapping the samples as in cross-fitting \citep[e.g.,][]{zheng2010asymptotic, kennedy2016semiparametric, Chernozhukov17, newey2018cross}. Following previous studies \citep[e.g.,][]{kennedy2020optimal, kennedy2021semiparametric}, for simplicity in the exposition we use a single split procedure in our analysis, as the extension to averages across independent splits is straightforward.
\end{remark}

In what follows, we present conditions under which the proposed estimators for $m^a$ and $\Sigma^a$, based on \eqref{eqn:return-estimator} and \eqref{eqn:cov-estimator}, are $\sqrt{n}$-consistent, asymptotically normal, and semiparametrically efficient.

\begin{lemma} \label{lem:double-robust}
Let $\widehat{m}^a$ and $\widehat{\Sigma}^a$ denote the estimators whose components are defined by \eqref{eqn:return-estimator} and \eqref{eqn:cov-estimator}, respectively. Under Assumptions \ref{assumption:B1} - \ref{assumption:B3}, we have
\begin{align*}
    \Vert \widehat{m}^a - m^a \Vert_2 &= O_\Pb\left(\Vert \widehat{\pi}_a - \pi_a \Vert_{2,\Pb} \max_{i} \Vert \widehat{\mu}_i - \mu_i \Vert_{2,\Pb} \right) + O_\Pb\left(n^{-1/2}\right),\\
    \Vert \widehat{\Sigma}^a - \Sigma^a \Vert_2 &= 
    O_\Pb\left(\Vert \widehat{\pi}_a - \pi_a \Vert_{2,\Pb} \left\{ \max_{i} \Vert \widehat{\mu}_i - \mu_i \Vert_{2,\Pb} +\max_{i,j} \Vert \widehat{\sigma}_{ij} - \sigma_{ij} \Vert_{2,\Pb} \right\} \right)+O_\Pb\left(n^{-1/2}\right).
\end{align*}
If we further assume the nonparametric conditions \ref{assumption:B4}, then
\begin{align} 
    \sqrt{n}(\widehat{m}^a_i - m^a_i) & \xrightarrow[]{d} N(0,\var(\phi^a_i)), \label{eqn:asymptotic-m} \\
    \sqrt{n}\left( \widehat{\Sigma}^a_{ij} - {\Sigma}^a_{ij} \right) 
    & \xrightarrow[]{d} [-\psi^a_j, -\psi^a_i, 1] N\left(0, \cov\left([{\phi}_{i}, {\phi}_{j}, {\phi}_{ij}]^\top\right) \right), \label{eqn:asymptotic-sigma}
\end{align}
and $\widehat{m}^a,  \widehat{\Sigma}^a$ are efficient, meaning that there exist no other regular estimators that are asymptotically unbiased and have smaller variance.
\end{lemma}

A proof of Lemma \ref{lem:double-robust}, along with the proofs of all other results, is provided in the supplementary material.
The result in Lemma \ref{lem:double-robust} is essentially due to the fact that our estimators are built from the efficient influence function, leading to second-order products of nuisance estimation errors.

Having established efficient estimation strategies for the counterfactual components of the optimization problem \ref{eqn:couterfactual-MV}, we now turn to the estimation and inference of its optimal solutions. Let $\widehat{w}$ be an optimal solution of the \emph{approximating program} of \ref{eqn:couterfactual-MV} in which we replace $m^a$, $\Sigma^a$ with their estimates $\widehat{m}^a$, $\widehat{\Sigma}^a$, respectively. Then $\widehat{w}$ is our proposed estimator for the optimal solution to \ref{eqn:couterfactual-MV}. Let $r_n$ be any sequence such that
\begin{align} \label{eqn:nuisance-product-rate}
    \Vert \widehat{\pi}_a - \pi_a \Vert_{2,\Pb} \left\{ \max_j \Vert \widehat{\mu}_j - \mu_j \Vert_{2,\Pb} + \max_{i,j} \Vert \widehat{\sigma}_{ij} - \sigma_{ij} \Vert_{2,\Pb} \right\} = O_\Pb\left( r_n \right).
\end{align}
The following result provides the rates of convergence for $\widehat{w}$, which follows directly from Theorem 2 in \citet{kim2025semiparametric}.
\begin{theorem} \label{cor:rates-opt-sol-MV}
Assume that $\Sigma^a$ is positive definite and let $w^* \equiv \sol(\text{\ref{eqn:couterfactual-MV}})$. Then under Assumptions \ref{assumption:B1}, \ref{assumption:B2},  \ref{assumption:B3}, 
\[
\Vert \widehat{w} - w^* \Vert_2 = O_\Pb\left(r_n \vee n^{-1/2}\right).
\]
If we additionally assume the second-order nonparametric conditions in \ref{assumption:B4}, then this becomes 
\[
\Vert \widehat{w} - w^* \Vert_2 = O_\Pb\left(n^{-1/2}\right).
\]
\end{theorem}

Theorem \ref{cor:rates-opt-sol-MV} shows that the proposed estimator can attain fast $\sqrt{n}$ rates even when we estimate all the nuisance regression functions at much slower rates; for example, it suffices that all the nuisance functions converge to their true values at a faster-than-$n^{1/4}$ rate in $L_2(\Pb)$ norm. This enables the use of a broad class of nonparametric regression techniques, depending on structural assumptions such as sparsity or smoothness \citep[][Section 4]{kennedy2016semiparametric}.

Since inference for optimal solution estimators is commonly performed using bootstrap methods, the case in which $\sqrt{n}(\widehat{w} - w^*)$ converges in distribution to a multivariate normal random vector is of particular importance. In the absence of this guarantee, the bootstrap procedure may yield inconsistent inference for the solution estimators \citep{fang2019inference}. Establishing the asymptotic distribution of $\widehat{w}$ requires stronger assumptions than those needed to ensure consistency.

For any feasible point $\bar{w} \in \mathcal{S}^a$ in \ref{eqn:couterfactual-MV}, we let
\begin{align*}
L(\bar{w},\bar{\gamma}) = {1}/{2}\bar{w}^\top\Sigma^a \bar{w}  - \lambda \bar{w}^\top m^a + \underset{j \in J_0(\bar{w})}{\sum}\bar{\gamma}_j \left( {B^a_j}^\top \bar{w} - c^a_j \right)
\end{align*}
denote the corresponding Lagrangian function with multipliers $\bar{\gamma} \geq 0$, and define the \emph{active index set} by
\[
J_0(\bar{w}) = \{1\leq j \leq r : {B^a_j}^\top \bar{w} - c^a_j = 0 \}.
\]
Then we require the following regularity condition.
\begin{enumerate}[label={}, leftmargin=*]
	    \item \customlabel{assumption:B5}{(B5)} 
For $w^* \in \sol({\text{\ref{eqn:couterfactual-MV}}})$ and the corresponding multipliers ${\gamma}^*$, we assume that $\{{B^a_j}^\top: j \in J_0(w^*)\}$ are linearly independent, and that the KKT conditions
\begin{align*}
    \Sigma^a w^*  - \lambda m^a + \underset{j \in J_0(w^*)}{\sum}{\gamma}^*_j {B^a_j}^\top w^* = 0, \quad \diag({\gamma}^*)(B^a w^* - c^a) = 0
\end{align*}
are satisfied such that
$
    {\gamma}^*_j > 0, \forall j \in J_0(w^*).
$    
\end{enumerate}	  

Assumption \ref{assumption:B5} ensures that the Linear Independence Constraint Qualification (LICQ) and Strict Complementarity (SC) hold at $w^* \in \sol({\text{\ref{eqn:couterfactual-MV}}})$. LICQ guarantees the validity of first-order KKT conditions at optimal solutions, while SC requires strictly positive dual variables for active constraints. Both conditions are commonly imposed to ensure well-posedness and tractability in nonlinear programming \citep[e.g.,][]{still2018lectures}. The following result provides sufficient conditions for establishing both $\sqrt{n}$-consistency and asymptotic normality of $\widehat{w}$.

\begin{theorem} \label{cor:asymptotics-opt-sol-MV}
Suppose that $\Sigma^a$ is positive definite and Assumptions \ref{assumption:B1}-\ref{assumption:B5} hold. For matrices $\textsf{C}_{ac} = \left[B^a_j, \, j \in J_0(w^*) \right]$ and $\gamma^*_{ac} = [\gamma^*_j, j \in J_0(w^*)]$, 
\begin{align} \label{eqn:closed-form-asymptotics}
    n^{\frac{1}{2}} \left(\widehat{w} - w^* \right) \xrightarrow{d} \begin{bmatrix}
        \Sigma^a & \textsf{C}_{ac}^\top \\
        \textsf{C}_{ac} & 0
        \end{bmatrix}^{-1} \begin{bmatrix}
        \bm{1} \\
        \diag(\gamma^*_{ac})\bm{1} 
        \end{bmatrix}^\top
        \mathcal{Z}_{w^*},
\end{align}
where $\mathcal{Z}_{w^*}$ is a mean-zero multivariate normal random vector such that
\begin{align*}
        n^{1/2}\begin{bmatrix}
        (\widehat{\Sigma}^a  - \Sigma^a) w^* + \sum_j\gamma_{j \in J_0(w^*)}^*\left\{ \widehat{B}^a_j - B^a_j \right\} \\
        -(\widehat{\textsf{C}}_{ac} - \textsf{C}_{ac})w^*
        \end{bmatrix} \xrightarrow{d} \mathcal{Z}_{w^*}.
\end{align*}
\end{theorem}

The above result is derived using Theorem 3 of \citet{kim2025semiparametric}.
Theorem \ref{cor:asymptotics-opt-sol-MV} implies that asymptotically valid confidence intervals and hypothesis tests can be constructed using the bootstrap method.

\section{Calibration of $\widehat{\Sigma}^a$} \label{sec:calibration-Sigma}

In the previous section, we introduced a semiparametric estimator for the counterfactual covariance matrix $\widehat{\Sigma}^a$, which serves as the quadratic component of the approximating program for \ref{eqn:couterfactual-MV}. However, in the absence of structural assumptions on the dependence (e.g., diagonality or factor models), the estimated covariance matrix may be ill-conditioned or fail to be positive (semi)definite. In our setting, such issues can substantially compromise the accuracy of the resulting optimal solution estimates. In this section, we present two calibration methods for the counterfactual covariance estimator that mitigate these challenges while preserving the convergence rate of the optimal solution estimator.

\subsection{Optimal Linear Shrinkage Estimation} \label{subsec:counterfactual-shrinkage}

If $\widehat{\Sigma}^a$ is ill-conditioned or nearly singular, solving linear systems involving $\widehat{\Sigma}^a$ becomes highly susceptible to numerical instability. As a result, each iteration of standard algorithms used to solve our approximating program may incur substantial numerical errors, potentially causing the algorithm to diverge by disrupting the descent direction. These issues can significantly impair the finite-sample accuracy of the estimated optimal solutions, despite the favorable asymptotic properties of the proposed estimator.

Covariance shrinkage offers a promising solution to this issue \citep[e.g.,][]{yang1994estimation, daniels1999nonconjugate, daniels2001shrinkage, ledoit2020power}. The core idea is to balance bias and variance by shrinking $\widehat{\Sigma}^a$ toward a target matrix, often interpreted as a reference prior. In this subsection, we develop a linear shrinkage estimator for ${\Sigma}^a$ by adapting the method of \citet{ledoit2004well}, one of the most widely used approaches for regularizing sample covariance matrices. However, adapting their method to our setting is nontrivial, as counterfactual outcomes are unobserved, rendering the direct use of sample covariance infeasible.

Let $\mathbb{I}$ denote the identity matrix. Our goal is to find the optimal linear combination of $\mathbb{I}$ and $\widehat{\Sigma}^a$ with minimum expected quadratic loss, which is represented by the solution of the following program\footnote{{Note that $\forall i,j$, $\widehat{\Sigma}^a_{ij}$ depends on the nuisance estimates $\widehat{\eta}_i, \widehat{\eta}_j, \widehat{\eta}_{ij}$, each of which is a function of {a separate independent sample} $\mathsf{D}_0^n$. {(See Assumption \ref{assumption:B1}.)} So in our notation, $\Pb\Vert \Sigma_S - \Sigma^a \Vert_F^2 = \E\left[ \Vert \Sigma_S - \Sigma^a \Vert_F^2 \mid \mathsf{D}_0^n \right]$.}}:
\begin{equation}
\label{eqn:opt-shrinkage}
\begin{aligned}
    & \underset{\rho, \nu \in \mathbb{R}}{\text{minimize}} \quad  \Pb\Vert \Sigma_S - \Sigma^a \Vert_F^2 \\
    & \text{subject to } \quad \Sigma_S= \rho\nu\mathbb{I} + (1-\rho)\widehat{\Sigma}^a.
\end{aligned}    
\end{equation}

Let ${\Sigma}^*_S \coloneqq \rho^*\nu^*\mathbb{I} + (1-\rho^*)\widehat{\Sigma}^a$ where $(\rho^*, \nu^*)$ is the optimal solution of \eqref{eqn:opt-shrinkage}. $\Sigma^*_S$ can be regarded as an oracle estimator that reduces the expected error of $\widehat{\Sigma}^a$ in the Frobenius norm (conditional on the nuisance parameter estimates) by shrinking it toward the matrix $\nu^*\mathbb{I}$. It is an oracle in the sense that the optimal shrinkage parameters $\rho^*$ and $\nu^*$ are unknown. In parallel to \cite{ledoit2004well}, we propose to estimate $\rho^*$ and $\nu^*$ by \cmmnt{\kim{If $\widehat{\Sigma}^a$ is relatively accurate so $\Pb\Vert \widehat{\Sigma}^a  - \Sigma^a \Vert_F^2$ is small enough, there is little need to shrink it; otherwise, we can improve the accuracy to a large extent by shrinking it toward the shrinkage target  $\nu^*\mathbb{I}$ with the optimal shrinkage intensity $\rho^*$.} Since computing $(\rho^*, \nu^*)$ requires knowledge of an unknown true parameter $\Sigma^a$, similarly to \cite{ledoit2004well}, we propose to estimate $\rho^*, \nu^*$ by}
$\widehat{\rho} = \frac{\widehat{\beta}^2}{\widehat{\delta}^2}$ and $\widehat{\nu}  = \frac{1}{k}\sum_{i=1}^k \widehat{\Sigma}^a_{ii}$, respectively, where $\widehat{\delta}^2 = \Vert \widehat{\Sigma}^a - \widehat{\nu}\mathbb{I} \Vert_F^2$ and $\widehat{\beta}^2=\min\{\tilde{\beta}^2, \widehat{\delta}^2 \} \quad \text{with} \quad \tilde{\beta}^2 = \frac{1}{n^2}\sum_{t=1}^n \Vert \widetilde{\Sigma}_{t}- \widehat{\Sigma}^a \Vert_F^2$, and the $(i,j)$-entry of the matrix $\widetilde{\Sigma}_{t}$ is defined by $\widehat{\phi}^a_{ij}(Z_{t})  - \Pn\{\widehat{\phi}^a_i(Z) \}\Pn\{\widehat{\phi}^a_j(Z)\}$, $1\leq \ i,j \leq k$, $1 \leq \ t \leq n$. Consequently, our proposed estimator for the optimal linear shrinkage is given by
\begin{align} \label{eqn:cov-shrinkage-est}
    \widehat{\Sigma}^*_S= \widehat{\rho}\widehat{\nu}\mathbb{I} + (1-\widehat{\rho})\widehat{\Sigma}^a.
\end{align}
The following theorem establishes the consistency of $\widehat{\Sigma}^*_S$.

\begin{theorem} \label{thm:consistency-cov-shrinkage}
Assume that \ref{assumption:B1} - \ref{assumption:B3} hold and that we have an initial estimate $\widehat{\Sigma}^a$ via \eqref{eqn:cov-estimator}.
Suppose further that $\mathbb{E}\left[\|Y_i^a\|^4\right] < \infty$ for all $i = 1, \ldots, k$.
Recall that $r_n$ is the rate given in \eqref{eqn:nuisance-product-rate} and, let $\widehat{w}_{S}$ denote the corresponding estimate for the optimal solution derived by substituting $\widehat{\Sigma}^*_S$ for $\widehat{\Sigma}^a$ in our approximating program. 
Then,
\begin{align*}
 \Vert \widehat{\Sigma}^*_S - {\Sigma}^*_S \Vert_F = O_\Pb\left(n^{-1/2} \vee r_n\right), \quad
\Vert \widehat{\Sigma}^*_S - \widehat{\Sigma}^a \Vert_F = O_\Pb\left(n^{-1/2} \vee r_n\right),
\end{align*}
and
\[
\left\Vert \widehat{w}_{S}- w^* \right\Vert_2 = O_\Pb\left(n^{-1/2} \vee r_n \right).
\]
\end{theorem}

To the best of our knowledge, although we focus on the case of fixed and finite $k$, the proposed estimator $\widehat{\Sigma}^*_S$ in \eqref{eqn:cov-shrinkage-est} is the first attempt to apply the idea of shrinkage estimation to counterfactual covariance matrices. 
$\widehat{\Sigma}^*_S$ is guaranteed to be non-singular. 
The use of $\widehat{\Sigma}^*_S$ mitigates the limitations of $\widehat{\Sigma}^a$ discussed above, and Theorem \ref{thm:consistency-cov-shrinkage} establishes that this adjustment does not alter the convergence rate of the optimal solution estimator $\widehat{w}$. 


\subsection{Positive Definite Correction}
A key assumption underlying Theorem \ref{cor:asymptotics-opt-sol-MV} is that ${\Sigma}^a$ is positive definite (PD). However, in practice, there is no guarantee that the estimator $\widehat{\Sigma}^a$ is positive definite as well, which implies that the resulting approximating program may fail to be (strictly) convex. This poses practical challenges: standard optimization methods may become trapped in local optima, and the lack of strict convexity prevents the use of efficient quadratic programming solvers \citep[e.g.,][]{osqp2020}, which are particularly important in large-scale applications.

PD correction methods, which replace $\widehat{\Sigma}^a$ with a nearby PD approximation $\widehat{\Sigma}^*_{\text{cor}}$, can be employed to address this issue. The following theorem establishes that such corrections do not affect the convergence rates established in earlier results, provided that $\widehat{\Sigma}^*_{cor}$ and $\widehat{\Sigma}^a$ get arbitrarily close in probability at a sufficiently fast rate. For instance, $\widehat{\Sigma}^*_{\text{cor}}$ may be obtained using the algorithm of \citet{higham2002accuracy}, which computes the nearest positive definite matrix. Alternatively, one may construct $\widehat{\Sigma}^*_{\text{cor}}$ using the model-free calibration procedure proposed by \citet{huang2017calibration}, incorporating a minimum eigenvalue threshold that vanishes as $n \to \infty$.

\begin{theorem} \label{thm:pd-correction}
Assume that $\Sigma^a$ is PD and let $\widehat{\Sigma}^*_{cor}$ denote a symmetric PD matrix indexed by $\widehat{\Sigma}^a$ such that $\Vert \widehat{\Sigma}^*_{cor} - \widehat{\Sigma}^a\Vert_2 = o_\Pb(1)$, and let $\widehat{w}_{cor}$ be the corresponding estimate for the optimal solution derived by substituting $\widehat{\Sigma}^*_{cor}$ for $\widehat{\Sigma}^a$ in our approximating program. Then under Assumptions \ref{assumption:B1} - \ref{assumption:B3}, 
\[
\left\Vert \widehat{w}_{cor}- w^* \right\Vert_2 =  O_\Pb\left(\Vert\widehat{\Sigma}^*_{cor} - \widehat{\Sigma}^a\Vert_2 \vee r_n\vee n^{-1/2} \right).
\]
If $\widehat{\Sigma}^*_{cor} = \widehat{\Sigma}^a$ whenever $\widehat{\Sigma}^a$ is PD, meaning that $\widehat{\Sigma}^a$ is only replaced when it is not PD, then the righthand side simplifies to $O_\Pb\left(r_n\vee n^{-1/2} \right)$.
\end{theorem}

As with the shrinkage method discussed in the previous section, PD corrections that satisfy the conditions of Theorem \ref{thm:pd-correction} offer a principled approach to enhancing the stability of the covariance matrix estimator. In particular, when (strict or strong) convexity is desired, such corrections can improve the finite-sample performance of the optimal solution estimator.
We note that more general shrinkage estimators, such as linear shrinkage estimators with custom designed targets or nonlinear shrinkage approaches, may also be adapted to our setting \citep[see, e.g.,][]{ledoit2020power}. However, extending our framework to accommodate these more sophisticated estimators lies beyond the scope of the present paper and is left for future work.

\section{Empirical Studies} \label{sec:numerical-illustraction}
In this section, we evaluate the performance of the proposed estimator on two simulated datasets and demonstrate its applicability through two real-world case studies. %

\subsection{\citeauthor{kang2007demystifying}'s Study with Multivariate Outcomes} \label{sec:simulation}
To estimate the counterfactual means and covariance matrices, we employ three methods: the PI, IPW, and our proposed SP estimator. Estimator performance is assessed via integrated bias and root-mean-squared error (RMSE), defined by
\[
\text{bias} = \frac{1}{k}\sum_{i=1}^k\Big\vert \frac{1}{B}\sum_{j=1}^B \widehat{w}^j_{i} - w^*_i \Big\vert, \quad \text{RMSE} = \frac{1}{k}\sum_{i=1}^k\left[ \frac{1}{B}\sum_{s=1}^B \left\{ \widehat{w}^s_{i} - w^*_i \right\}^2 \right]^{1/2},
\]
across $B=250$ simulations, where $w^*$ are the optimal weights of the true QP and $\widehat{w}^s$ are estimates of $w^*$ at the $s$-th simulation. Our data generation is based on the simulation study by \citet{kang2007demystifying}, modified to accommodate multivariate outcomes as follows:
\begin{gather*}
    X=(X_1,X_2,X_3,X_4) \sim N(0,I),\\
    \pi_a(X) = \text{expit}(-0.5X_1 + 0.25X_2 - 0.125X_3 - 0.05X_4),\\
    (Y_i \mid X,A) \sim N(\mu_i(X,A),V_i),
\end{gather*}
where for $i=1,...,k$,
\begin{gather*}
\mu_i(X,A) = b_i + A(d_i + X\alpha),\\
b_i \sim \Unif(-0.5, 0.5), \quad d_i \sim \Unif(-0.25, 0.25), \\
\alpha = \left(0.1 + u_{i1}, -0.1 + u_{i2}, 0.2 + u_{i3}, -0.2 + u_{i4} \right)^\top,\\
u_{i1},u_{i2},u_{i3},u_{i4} \sim \Unif(-0.5, 0.5),\\
V_i \sim \Unif(1.5, 3).
\end{gather*}

\begin{figure}[t!]
    \centering
    \subfigure{\includegraphics[width=0.48\textwidth]{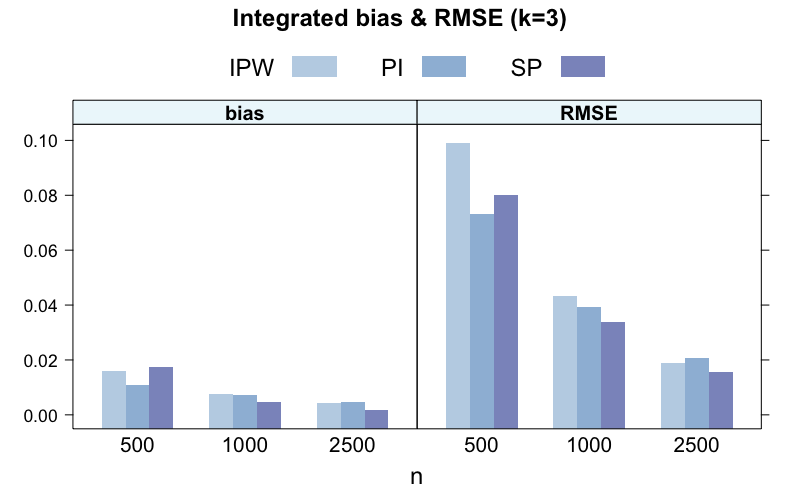}} 
    \hfill%
    \subfigure{\includegraphics[width=0.48\textwidth]{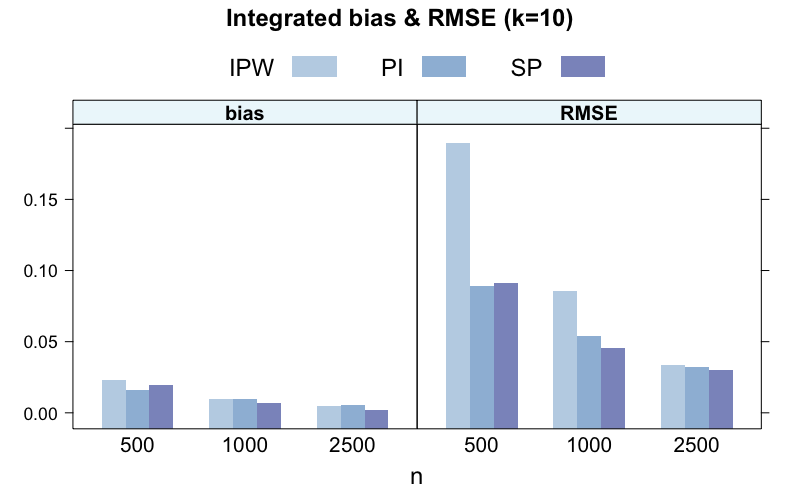}}
    \caption{Finite sample performance of the three estimators based on non-transformed covariates $X$, across different sample sizes.}
    \label{fig:sim1-cor}
\end{figure}
\begin{figure}[t!]
    \centering
    \subfigure{\includegraphics[width=0.48\textwidth]{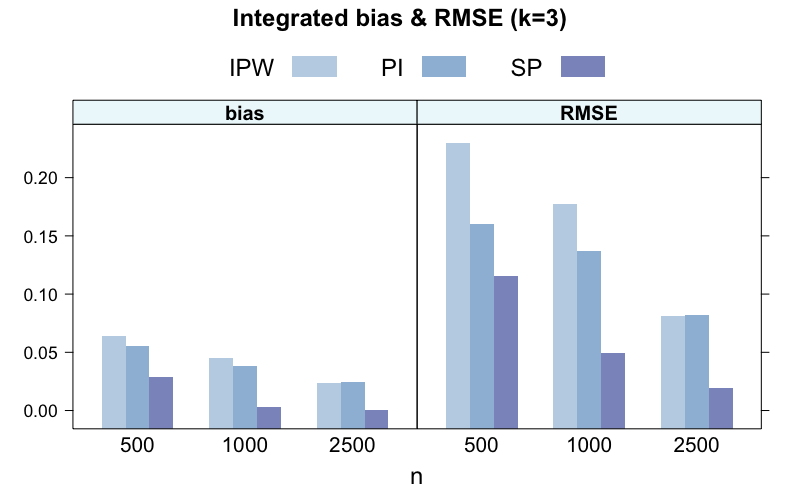}}
    \hfill%
    \subfigure{\includegraphics[width=0.48\textwidth]{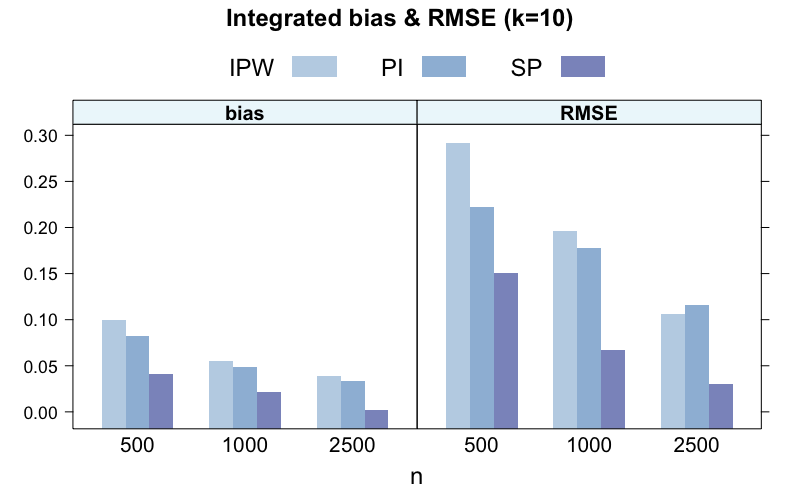}}
    \caption{Finite sample performance of the three estimators with the transformed covariates. At each simulation, either the propensity score ($\pi$) or outcome regression models ($\mu, \sigma$) is estimated using the transformed covariates $\widetilde{X}$, not $X$.}
    \label{fig:sim1-mis}
\end{figure}

Here, $\Unif(l, u)$ denotes the uniform distribution over the interval $[l,u]$. We use sample sizes $n=500, 1000, 2500$. Throughout this section, to estimate the nuisance regression functions, we use the cross-validation super learner ensemble estimator implemented in the \texttt{SuperLearner} R package to combine generalized additive models, multivariate adaptive regression splines, and random forests. For simplicity, we estimate all the nuisance functions on a separate independent set with equal sample size, and the minimum level of the weighted mean outcome $\text{r}_{min}$ is set to $-\infty$. 

We analyze the counterfactual regime corresponding to $A = 1$.
We consider two versions for each of the three estimators, depending on how each of the nuisance functions are estimated: using the baseline covariates $X$ or using transformed covariates $\widetilde{X}$, based on the same transformations as in \citet{kang2007demystifying}, i.e.,
\[
\widetilde{X}=\left(\exp(X_1/2), X_2/(1+\exp(X_1)) + 10, (X_1X_3/25+0.6)^3, (X_2+X_4+20)^2 \right).
\]
When the transformed covariates $\widetilde{X}$ are used, estimation of the nuisance functions is more challenging, and at each round of simulation we estimate either the propensity score ($\pi_a$) or outcome regressions ($\mu_i, \sigma_{ij}$) using $\widetilde{X}$ with equal chance. In other words, $\widetilde{X}$ is used to estimate $\pi_a$ for roughly $B/2$ simulations and to estimate ($\mu_i, \sigma_{ij}$) for the remaining simulations. 

Our results use the shrinkage estimator $\widehat{\Sigma}^*_S$ developed in Section \ref{subsec:counterfactual-shrinkage} as it shows a slight improvement in RMSE than the PD correction method. In general, we achieve between $10$ and $30$ percent relative improvement with the proposed calibration methods (see Appendix \ref{appsec:calibration-additional} for details).
The results are presented in Figures \ref{fig:sim1-cor}, \ref{fig:sim1-mis}.

In Figure \ref{fig:sim1-cor}, the proposed estimator performs as well or slightly better than the PI or IPW estimators. However, in Figure \ref{fig:sim1-mis}, when one of the nuisance estimators is based on  $\widetilde{X}$, the proposed estimator gives substantially smaller bias and RMSE in general, and performs better with $n$ than do the other methods. This interesting behavior follows from the results in Section \ref{sec:estimation-inference} that the proposed estimator has second-order multiplicative bias and thus it is sufficient to require $n^{1/4}$ rates on nuisance estimation in order for this estimator to attain $\sqrt{n}$ rates, while the PI and IPW directly inherit the slower-than-$\sqrt{n}$ rates at which the nuisance parameters are estimated and are expected to be converge particularly slowly when $\widetilde{X}$ is used. This behavior appears to hold regardless of the value of the decision variable $k$, although we have slightly larger bias and RMSE for $k=10$ than $k=3$. 

\subsection{Relative Improvement in RMSE through Covariance Calibration}
\label{appsec:calibration-additional}

\begin{figure}[t!]
    \centering
    \subfigure{\includegraphics[width=0.48\textwidth]{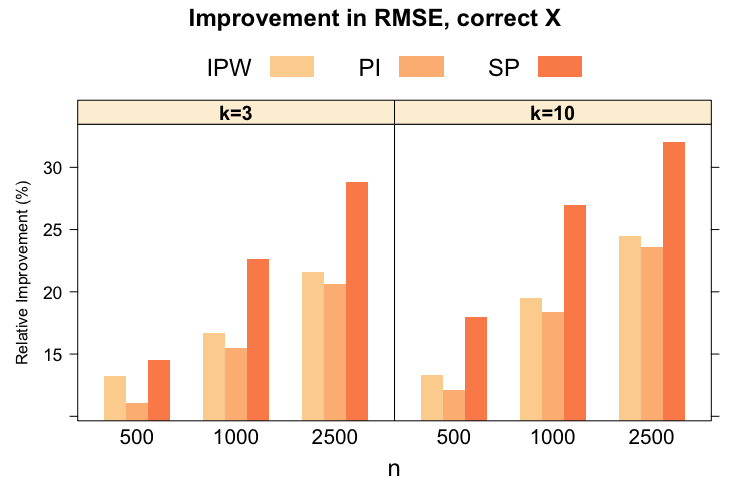}} 
    \hfill%
    \subfigure{\includegraphics[width=0.48\textwidth]{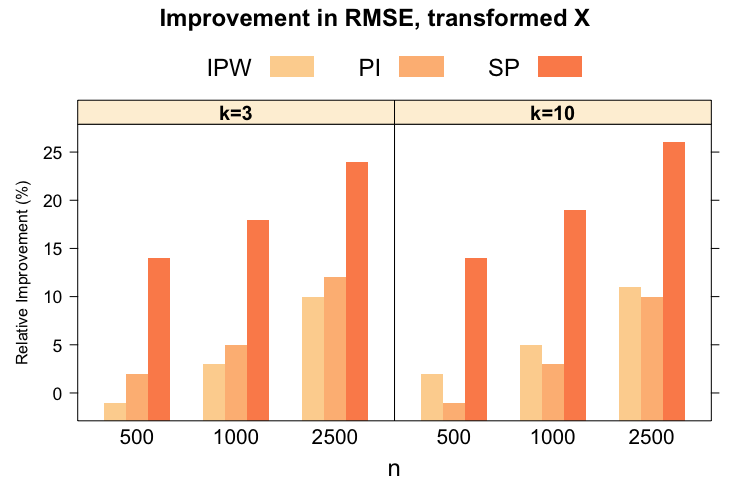}}
    \caption{Relative improvement in RMSE of the three estimators through the covariance shrinkage, based on the non-transformed covariates $X$ (left) and the transformed covariates (right).}
    \label{fig:rel-imp-linear-shrinkage}
\end{figure}
\begin{figure}[t!]
    \centering
    \subfigure{\includegraphics[width=0.48\textwidth]{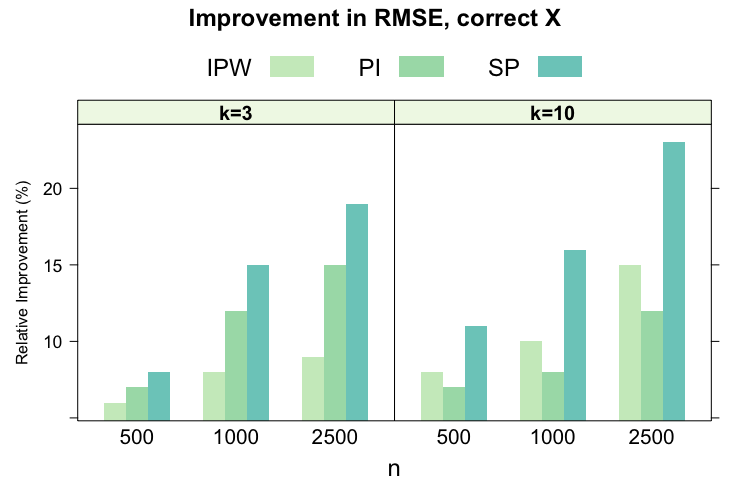}}
    \hfill%
    \subfigure{\includegraphics[width=0.48\textwidth]{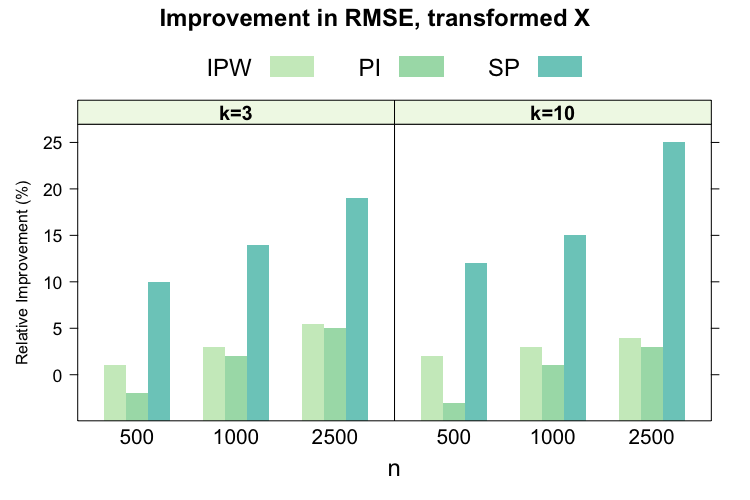}}
    \caption{Relative improvement in RMSE of the three estimators through the PD correction, based on the non-transformed covariates $X$ (left) and the transformed covariates (right).}
    \label{fig:rel-imp-psd-correction}
\end{figure}

We conduct an additional simulation, using the same setup as in Section \ref{sec:simulation}, to demonstrate that our proposed calibration methods enhance the performance of the proposed optimal solution estimator relative to that based on $\widehat{\Sigma}^a$. For each calibration method $\widehat{\Sigma}^* \in \{\widehat{\Sigma}^*_S, \widehat{\Sigma}^*_{cor}\}$, we compute the percentage relative improvement in RMSE using the following formula
\begin{align*}
    \frac{\left\{\text{RMSE}(\widehat{\Sigma}^a) - \text{RMSE}(\widehat{\Sigma}^*)\right\}}{\text{RMSE}(\widehat{\Sigma}^*)},
\end{align*}
where the RMSE of the optimal solution estimator $\widehat{w}$ is computed in the same way as Section \ref{sec:simulation} with $\widehat{\Sigma}^a$ or its calibrated version $\widehat{\Sigma}^*$. Again, we construct all the nuisance estimators on the independent, separate sample with the same size. For $\widehat{\Sigma}^*_{cor}$ we compute the the nearest PD matrix by using R function \texttt{nearPD}, which implements the algorithm of \citet{higham2002accuracy}, and then forces positive definiteness if needed.

The results are presented in Figures \ref{fig:rel-imp-linear-shrinkage} and \ref{fig:rel-imp-psd-correction}. Both calibration methods appear to significantly improve upon the original optimal solution estimator that is computed without covariance calibration, although the improvement becomes less substantial as sample size increases. More importantly, without calibration, the number of suboptimal solutions, i.e., solutions that fail to converge, increases noticeably. Among the three estimators, the semiparametric estimator yields the greatest improvement. When the transformed covariates are used, the relative improvement in RMSE has been largely wiped out for the PI and IPW estimators; however, the semiparametric estimator continues to exhibit substantial gains. In general, larger improvements are observed for the shrinkage estimator than the PD correction with the given simulation setup.

\subsection{Optimal Medical Appointment Scheduling} \label{subsec:medical-appt}
Here, we present the first case study demonstrating the practical applicability of the proposed methods.
Medical providers have finite time to provide care for large populations of patients. In order to accommodate patients' scheduling needs and their own staffing needs, providers must choose how many appointment slots to reserve for \emph{fixed} appointments, which are scheduled in advance, vs. \emph{open-access} appointments, which are scheduled on short notice, often the same day that patients request them. Providers naturally wish to maximize the daily utilization rate, i.e., the proportion of slots each day in which patients are actually seen, while minimizizing variance in this rate across days. The utilization rate depends in part on the patient no-show rate, which can be quite high for fixed appointments.

\citet{qu2012MeanVarianceModel} employed mean–variance optimization to determine the optimal allocation between fixed and open-access appointments across various provider types. \cmmnt{, as it is desirable for providers to fill as many appointment slots as possible while minimizing the variance (uncertainty) in the proportion of slots that are utilized. , where the outcome of interest is the number of patients seen per session.} Their study considered a simplified setting in which all relevant parameters were assumed to be known, eliminating the need for estimation from observed data. Building on their approach, we illustrate how our counterfactual framework can enable reliable decision support in healthcare, particularly in the presence of interventions that may substantially shift the distribution of outcomes.

We generate a simulated dataset of 10,000 observations describing patient appointments, appointment types, an intervention to improve patient attendance, and utilization rates as follows:
\begin{gather*}
    X \sim \Unif(-1, 1) \\
    \pi_1(X) = \expit(0.6 + 0.1*X^3) \\
    Y_o \mid A, X \sim \Beta(1 + A/5, X^2) \\
    Y_f \mid A, X \sim \Beta(0.1 + A/2, X^2)
\end{gather*}
with $Y_o \ind Y_f \mid A, X$. Here, the intervention $A \in \{0, 1\}$ represents two types of appointment reminders. These patient prompts have been shown to reduce no-show rates for fixed appointments, with calls from staff $(A = 1)$ leading to greater improvements than automated reminders $(A = 0)$ \citep{parikh2010EffectivenessOutpatientAppointment}. Suppose that medical providers are interested in determining the optimal mix of two appointment types under policies that assign $A = 0$ or $A = 1$ to all appointments, based on data in which $A$ varies across observations. For example, they may have been piloting an automated calling system ($A = 1$), or may be transitioning from infrequent manual reminder calls to a policy of making calls before every appointment. $Y_o$ and $Y_f$ represent the observed utilization rates, the proportions of daily open-access and fixed appointments, respectively, in which providers see patients. For simplicity, we assume that the utilization rates do not depend on the number of appointments of each type offered. $X$ represents a synthesis of variables that influence $A$, or $Y_o$ and $Y_f$, such as weather conditions, which can affect both staff availability and patient no-show rates, or the severity of patient comorbidities. 

Since open-access appointments are made on very short notice, it is likely that the benefit of personal reminders relative to automated ones is greater for fixed appointments than for open-access appointments. Our data generating process reflects this view, with Table \ref{tab:utilization_rates} showing that personal calls increase utilization rates substantially relative to automated ones for fixed appointments, but have only a small effect for open-access appointments.
\begin{table}[]
    \centering
    \begin{tabular}{ccc}
        $A$ & $Y_o$ & $Y_f$  \\
        \hline
        0 & 0.79 (0.08) & 0.68 (0.14) \\
        1 & 0.81 (0.07) & 0.79 (0.08)
    \end{tabular}
    \caption{Means (variances) in utilization rates for open-access $(Y_o)$ vs fixed $(Y_f)$ appointments.}
    \label{tab:utilization_rates}
\end{table}

\begin{figure}[!t]
    \centering
    \includegraphics[width=0.75\linewidth]{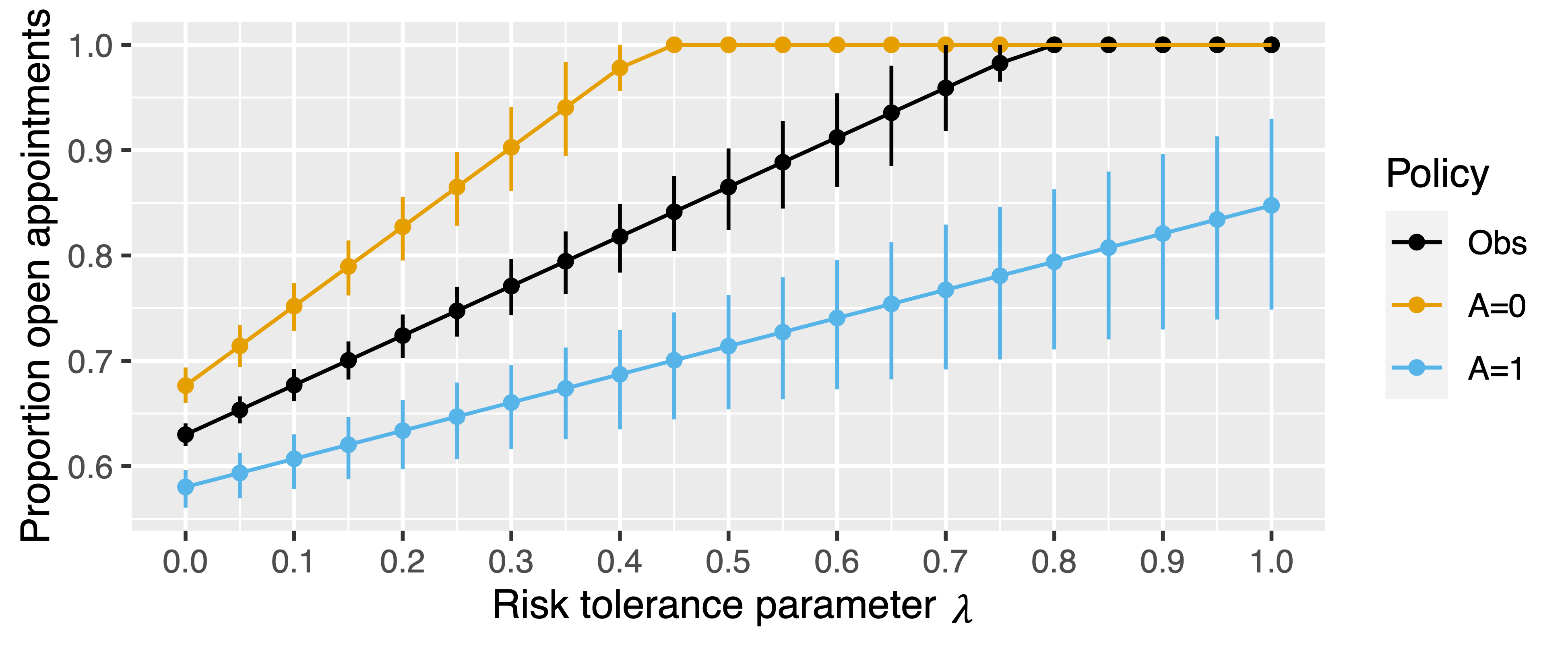}
    \caption{Estimated optimal proportion of open appointments under the two counterfactual policies and the observable policy (``Obs''), with 95\% bootstrap confidence intervals, across different risk tolerances.}
    \label{fig:proportion_open_appointments}
\end{figure}
\begin{figure}[!t]
    \centering
    \includegraphics[width=0.75\linewidth]{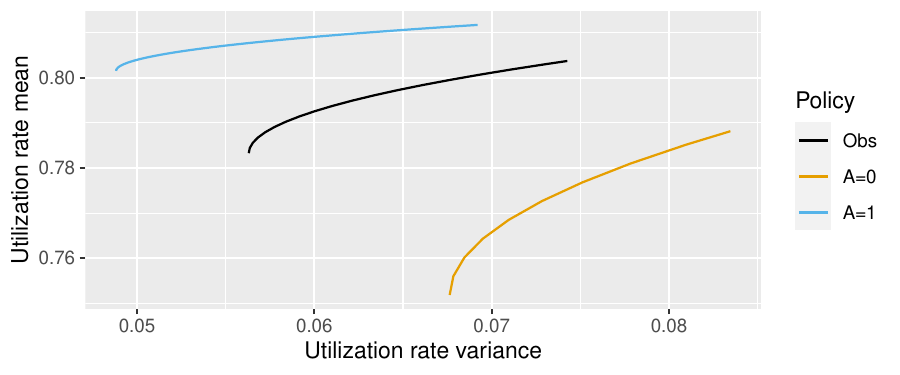}
    \caption{Pareto-efficient frontiers under the two counterfactual policies and the observable policy. Each point on a each curve represents an estimated optimum allocation of fixed vs. open-access appointments for a given value of the risk tolerance parameter $\lambda$.} 
    \label{fig:appointments_pareto}
\end{figure}

Figure \ref{fig:proportion_open_appointments} shows the estimated optimal proportion of open appointments under the two counterfactual conditions and the observable condition, across a range of values of the risk tolerance parameter $\lambda$. Because utilization is higher on average for open-access appointments, the optimal proportion increases as risk tolerance increases. Because personal calls ($A=1$) increase utilization for fixed appointments compared to automated ones ($A=0$), the optimal proportion of open-access appointments is smaller if personal calls were to be implemented globally than if automated calls were to be implemented globally, or if the provider were to continue with the current practice of automated reminder calls for some appointments or personally made for others.

Figure \ref{fig:appointments_pareto} shows the estimated counterfactual Pareto-efficient frontiers for the three conditions, with each curve spanning $\lambda \in [0, 2]$. The leftmost point on each curve represents $\lambda=0$, i.e., the composition of appointments that minimizes variance without regard to the mean. The $A=1$ curve dominates the others, meaning that for any given mean utilization rate (variance), the lowest variance (highest mean utilization rate) is achieved by the personal reminder policy. This suggests that the provider may achieve the best mean-variance tradeoff by choosing the personal reminder policy ($A = 1$). This example illustrates how the proposed methods support reliable healthcare decision-making by optimizing resource allocation under constraints, where each target policy may substantially shift the outcome distribution.

\subsection{Counterfactual Portfolio Modeling}
We next illustrate our method in the context of classical financial portfolio modeling. We consider monthly returns from six Vanguard index funds representing different asset classes, the same funds used in \citet{kim2021MeanVarianceOptimization}: U.S. large caps (VFIAX), U.S. small caps (VSMAX), developed markets outside the U.S. (VTMGX), emerging markets (VEIEX), the U.S. total bond market (VBTLX), and U.S. medium- and lower-quality corporate bonds (VWEAX). We use daily adjusted closing prices collected from Yahoo Finance.

The intervention considered in this example is the \emph{federal funds effective rate}, which represents the average interest rate at which banks lend to one another overnight. This rate is influenced by the {federal funds target rate}, which is set by the Federal Reserve (Fed). \cmmnt{\footnote{In recent years, the Fed has set a target range rather than a single target rate.} \cmmnt{The Federal Reserve influences the effective rate through actions such as the buying and selling of government bonds, but the effective rate ultimately arises from market forces, which means it may deviate from the target rate by a variable amount \citep{federalreserve2021FedExplained}.} There is a large literature devoted to understanding how changes in federal interest rates affect asset prices over various time spans \citep{bernanke2005WhatExplainsStock, li2010ImpactMonetaryPolicy, miranda-agrippino2020MonetaryPolicyGlobal}. Though the nature and timing of these effects is disputed \citep{bouakez2013StockReturnsMonetary, rigobon2004ImpactMonetaryPolicy}, it is likely that the optimal portfolio weights differ under different (counterfactual) rate environments.
}
There is a large literature devoted to understanding how changes in federal interest rates affect asset prices over various time spans \citep[e.g.,][]{li2010ImpactMonetaryPolicy, miranda-agrippino2020MonetaryPolicyGlobal}. Though the nature and timing of these effects is disputed, it is likely that the optimal portfolio weights differ under different (counterfactual) rate environments \citep[e.g.,][]{bouakez2013StockReturnsMonetary}.

For each month, we let $A=1$ if the effective rate increased with respect to the previous month, and $A=0$ otherwise. For example, if the effective rate for October was 3.0 and the effective rate for November was 3.25, then we would have $A=1$ for November. \cmmnt{\footnote{This reflects the fact that effective rates are market-determined, so the effective monthly rate is only realized at the end of the month.} Recall that we require a set of covariates that meet the no unmeasured confounding assumption (C2) and the positivity assumption (C3). One possibility is to include covariates that predict interest rates and/or (counterfactual) asset prices up to noise. }As covariates, we include the Consumer Price Index (CPI) and the unemployment rate, which correspond to the Fed's dual mandate to promote maximum employment and price stability \citep{federalreserve2021FedExplained}. We also include the five factors from the Fama and French asset pricing model, which aim to explain long-term expected portfolio returns \citep{fama2015FivefactorAssetPricing}. \cmmnt{Although there are many other macroeconomic variables that have historically been used to predict price movements and returns, it is debatable to what extent they have predictive value that is of use to investors \citep{welch2008ComprehensiveLookEmpirical}.} \cmmnt{In the interest of keeping the number of covariates small in order to minimize nuisance parameter estimation error, and since we primarily intend this as in illustrative exercise, we do not include any additional covariates.}

Our data span 2011-2020 (120 months), of which 58 months involve rate increases $(A = 1)$. \cmmnt{We randomly split the data into equally sized folds, estimate the nuisance parameters on one fold, and estimate the optimal weights on the other fold.} Figure \ref{fig:asset_weights} shows the estimated optimal weights for a range of values of the risk tolerance parameter $\lambda$, under the two counterfactual scenarios as well as the traditional observable setting. (We use the word ``scenarios'' to emphasize the fact that investors have no control over Fed policy.) \cmmnt{(We refer to these settings as ``scenarios'' rather than ``policies'' because unlike in the medical appointment application, the user does not control the treatment process; investors have no control over Fed policy.) \mishler{The ``observable'' weights are widely used in the context of robo-advising, which is a large and growing industry \citep{beketov2018RoboAdvisorsQuantitative}. Figure 5 additionally shows estimated weights calculated via traditional mean-variance optimization using only the data in which $A$ is equal to 0 (denoted ``$|A=0$'') or only the data in which $A$ is equal to 1 (``$|A=1$''). These cases represents the weights that would be ideal conditional on the Fed actually raising or not raising rates, rather than on the counterfactual scenario in which the Fed continually raises or does not raise rates.}}
As expected, for small values of $\lambda$, the portfolios all tilt heavily toward the U.S. total bond market (VBTLX), which has the lowest return and lowest volatility among the asset classes. As $\lambda$ increases, the portfolios tilt instead toward assets with higher return and higher volatility. In the $A=1$ scenario, the portfolio tilts toward corporate bonds (VWEAX), whereas in the other two scenarios, the portfolio tilts toward large cap stocks (VFIAX), suggesting that the volatility-return tradeoffs differ in a (counterfactual) environment in which rates are rising versus an environment in which they are steady or falling.

\begin{figure} 
    \centering
    \includegraphics[width=.65\linewidth]{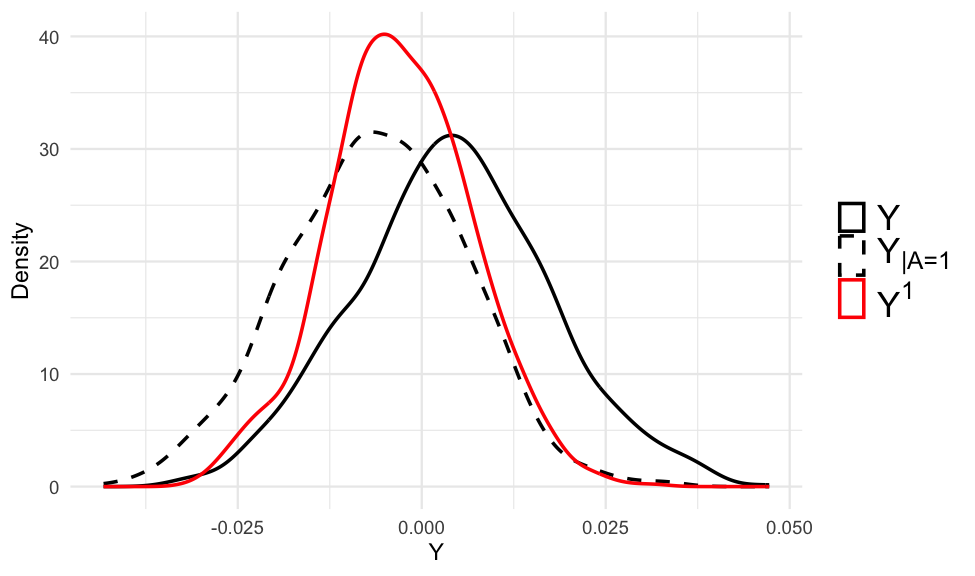}
    \caption{Distribution of the observed asset return for VWEAX (\(Y\)), the return conditional on \(A = 1\) using a subset of observed data (\(Y \mid A = 1\)), and the counterfactual return distribution under the policy \(A = 1\) (\(Y^1\)).}
    \label{fig:counterfactual-return-dist}
\end{figure}

We focus on counterfactual return distributions under $A = a$, i.e., $Y^a$, for $a \in \{0,1\}$, where potential confounders, such as the economic and firm-specific conditions surrounding the Fed’s decisions, as described above, are appropriately adjusted for, in order to isolate the effects attributable solely to the Fed’s action. This contrasts with $Y \mid A = a$, where shifts in the outcome distribution may be influenced by confounding (see Figure~\ref{fig:counterfactual-return-dist}). Using the proposed estimator, we compute the portfolio weights under counterfactual scenarios, where the mean and variance are calculated from the counterfactual returns $Y^0$ and $Y^1$. These are compared to weight estimates obtained from subsets of the observed data in which $A = 0$ or $A = 1$ (denoted by ``$|A=0$'' and ``$|A=1$'', respectively). The results are presented in Figure~\ref{fig:asset_weights}. There, we observe substantial differences between the counterfactual and factual optimal portfolios.

These counterfactual portfolios, which have never been studied in the literature, may be of intrinsic scientific interest. They can provide additional insight into portfolio robustness, by illuminating the sensitivity of the portfolio weights to surprise rate hikes or cuts. For example, consider an asset manager who believes that the current Fed is more aggressive toward inflation than previous Feds. Then the observational weights may be based on an underestimation of the Fed's likelihood of raising rates, and the asset manager may wish to tilt their portfolio toward the weights in the $A=1$ scenario. We present this as a heuristic argument for now, and leave a more thorough analysis of the uses of our framework in practice for future work.

\begin{figure} 
    \centering
    \includegraphics[width=.85\linewidth]{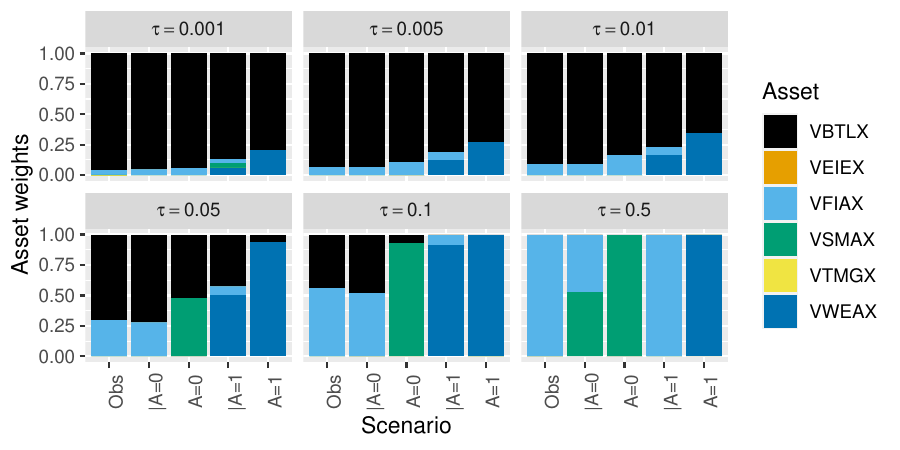}
    \caption{Estimated optimal asset weights in the observational and counterfactual scenarios.}
    \label{fig:asset_weights}
\end{figure}

\section{Discussion} \label{sec:discussion}

In this paper, we proposed counterfactual mean–variance optimization, a novel framework for determining optimal resource allocation under constraints, in the presence of hypothetical interventions that define unseen scenarios potentially far removed from the observed world.
Leveraging recent advances in counterfactual prediction, we developed a doubly robust estimator that achieves $\sqrt{n}$-consistency and asymptotic normality, even when employing flexible nonparametric regression methods. To address potential numerical instabilities, we also introduced calibration methods for the counterfactual covariance matrix estimator, which mitigate issues such as ill-conditioning and non-positive (semi)definiteness. Our methods were validated through simulation studies and demonstrated in real-world case studies in healthcare operations and financial portfolio optimization. Furthermore, the shrinkage approach proposed in Section~\ref{subsec:counterfactual-shrinkage} opens the door to extending a broad class of covariance matrix estimation techniques, such as those discussed in \citet[][Section 2.4]{ledoit2020power}, to counterfactual inference settings.

There are several promising avenues for future research. First, our methods may be extended to optimal resource allocation problrms under additional constraints, such as budget, fairness, or interpretability constraints, which frequently arise in domains like medicine and the social sciences. Second, the objective in \eqref{eqn:couterfactual-MV} could be generalized to incorporate richer reward and risk functionals beyond the mean and variance. For instance, to capture asymmetry in return distributions, one could employ Value-at-Risk by replacing $w^\top \Sigma^a w$ with $Q_{\alpha}(w^\top m^a)$, where $Q_{\alpha}(\cdot)$ denotes the $\alpha$-quantile. Lastly, while our current formulation assumes linear constraints, extending the framework to accommodate nonlinear constraints may open new insights into counterfactual resource allocation in complex decision environments.



\bibliographystyle{agsm}
\bibliography{bibliography}

\vspace*{\fill}
\paragraph{Disclaimer}
\begin{singlespace*}
{\scriptsize This paper was prepared for informational purposes by the Artificial Intelligence Research group of JPMorgan Chase \& Co. and its affiliates (``JP Morgan''), and is not a product of the Research Department of JP Morgan. JP Morgan makes no representation and warranty whatsoever and disclaims all liability, for the completeness, accuracy or reliability of the information contained herein. This document is not intended as investment research or investment advice, or a recommendation, offer or solicitation for the purchase or sale of any security, financial instrument, financial product or service, or to be used in any way for evaluating the merits of participating in any transaction, and shall not constitute a solicitation under any jurisdiction or to any person, if such solicitation under such jurisdiction or to such person would be unlawful.}
\end{singlespace*}

\newpage
\appendix
\input{appendix}

\end{document}

%% file: appendix.tex
\begin{center}
{\large\bf SUPPLEMENTARY MATERIAL}
\end{center}
\vspace*{.1in}
\setcounter{page}{1}
\spacingset{1.2}

\section{Proofs}

\textbf{Extra notation.} First, we introduce some extra notation used throughout in the proofs. We let $\langle M_1, M_2 \rangle \coloneqq tr\left(M_1^\top M_2\right)/k$ for $k\times k$ matrices $M_1, M_2$ (so $\Vert M_1 \Vert_F = \sqrt{\langle M_1, M_1 \rangle}$). We let $\mathbb{B}_{\delta}(\bar{z})$ denote the open ball with radius $\delta > 0$ around the point $\bar{z}$ with  $\Vert \cdot \Vert_2$ (unless otherwise mentioned), i.e., $\mathbb{B}_{\delta}(\bar{z})= \{z \mid \Vert z - \bar{z} \Vert_2 < \delta \}$. We use $C^r(\mathbb{S})$ to denote a set of functions that are $r$ times continuously differentiable on $\mathbb{S}$. 

\subsection{Proof of Lemma \ref{lem:double-robust}} \label{proof-prop-double-robust}
\begin{proof}
Recall that we have
\begin{align*}
    & \phi_i(Z;\eta_i) = \frac{\mathbbm{1}(A=a)}{\pi(X)}\left\{Y_i - \mu_i(X,A)\right\} + \mu_i(X,a), \\
    & \phi_{ij}(Z;\eta_{ij}) = \frac{\mathbbm{1}(A=a)}{\pi(X)}\left\{Y_iY_j - \sigma_{ij}(X,A)\right\} + \sigma_{ij}(X,a),
\end{align*}
as the uncentered efficient influence functions for the parameter ${\psi}_{i} = \E[Y^a_i] = \E\{\E[Y_i \mid X, A=a]\}$ and ${\psi}_{ij} = \E[Y^a_iY^a_j] = \E\{\E[Y_iY_j \mid X, A=a]\}$ with the relevant nuisance functions $\eta_i=\{\pi(X), \mu_i(X,A)\}$, $\eta_j=\{\pi(X), \sigma_{ij}(X,A)\}$, respectively.


Recall that our proposed estimators in \eqref{eqn:return-estimator}, \eqref{eqn:cov-estimator} are given by
\begin{align*}
    &\widehat{m}^a_{i} = \widehat{\psi}_{i},\\
    &\widehat{\Sigma}^a_{ij} = \widehat{\psi}_{ij} - \widehat{\psi}_{i}\widehat{\psi}_{j},
\end{align*}
where 
\begin{align*}
    &  \widehat{\psi}_{i} = \Pn\left\{\phi_i(Z;\widehat{\eta}_i) \right\}, \\
    &  \widehat{\psi}_{j} =\Pn\left\{\phi_j(Z;\widehat{\eta}_j)\right\}, \\
    &  \widehat{\psi}_{ij} = \Pn\left\{\phi_{ij}(Z;\widehat{\eta}_{ij})\right\}.
\end{align*}

$\widehat{\psi}_{ij}$, $\widehat{\psi}_{i}$, $\widehat{\psi}_{j}$ are semiparametric estimators for the mean outcomes ${\psi}_{ij} = \E[Y^a_iY^a_j]$, ${\psi}_{i} = \E[Y^a_i]$, ${\psi}_{j} = \E[Y^a_j]$. Hence with \eqref{assumption:B2}, Together with either the Donsker condition or sample splitting, it follows that by \citet{kennedy2016semiparametric},
\begin{align*}
    & \widehat{\psi}_{i} - {\psi}_{i} = (\Pn - \Pb)\phi_i(Z) +  O\left( \Vert \widehat{\pi}_a - \pi_a \Vert_{2,\Pb} \Vert \widehat{\mu}_i - \mu_i \Vert_{2,\Pb} \right) + O_\Pb\left( \frac{\Vert\widehat{\psi}_{i} - {\psi}_{i} \Vert}{\sqrt{n}}\right) , \\
    & \widehat{\psi}_{j} - {\psi}_{j} = (\Pn - \Pb)\phi_j(Z) + O\left( \Vert \widehat{\pi}_a - \pi_a \Vert_{2,\Pb} \Vert \widehat{\mu}_j - \mu_j \Vert_{2,\Pb} \right) + O_\Pb\left( \frac{\Vert\widehat{\psi}_{j} - {\psi}_{j} \Vert}{\sqrt{n}}\right), \\
     & \widehat{\psi}_{ij} - {\psi}_{ij} = (\Pn - \Pb)\phi_{ij}(Z) + O\left( \Vert \widehat{\pi}_a - \pi_a \Vert_{2,\Pb} \Vert \widehat{\sigma}_{ij} - \sigma_{ij} \Vert_{2,\Pb} \right) + O_\Pb\left( \frac{\Vert\widehat{\psi}_{ij} - {\psi}_{ij} \Vert}{\sqrt{n}}\right),
\end{align*}
and thus by the central limit theorem and the given consistency conditions,
\begin{align*}
    & \widehat{m}^a_{i} - {m}^a_{i} = O_\Pb\left( \Vert \widehat{\pi}_a - \pi_a \Vert_{2,\Pb} \Vert \widehat{\mu}_i - \mu_i \Vert_{2,\Pb} \right) +  O_\Pb\left(n^{-1/2}\right), \\
    & \widehat{\Sigma}^a_{ij} - {\Sigma}^a_{ij} = O_\Pb\left( \Vert \widehat{\pi}_a - \pi_a \Vert_{2,\Pb} \left\{ \Vert \widehat{\mu}_i - \mu_i \Vert_{2,\Pb} + \Vert \widehat{\mu}_j - \mu_j \Vert_{2,\Pb} + \Vert \widehat{\sigma}_{ij} - \sigma_{ij} \Vert_{2,\Pb} \right\} \right) +  O_\Pb\left(n^{-1/2}\right).
\end{align*}

Since $k$ is finite, we have
\begin{align*}
    \Vert \widehat{\Sigma}^a - \Sigma^a \Vert_2 &\leq \Vert \widehat{\Sigma}^a - \Sigma^a \Vert_F \\
    & = \left( \underset{i,j}{\sum} \left\vert \widehat{\Sigma}^a_{ij} - {\Sigma}^a_{ij} \right \vert^2 \right)^{1/2} \\
    & \leq \underset{i,j}{\sum} \left\vert \widehat{\Sigma}^a_{ij} - {\Sigma}^a_{ij} \right\vert \\
    & = O_\Pb\left(\Vert \widehat{\pi}_a - \pi_a \Vert_{2,\Pb} \left\{ \max_{i} \Vert \widehat{\mu}_i - \mu_i \Vert_{2,\Pb} +\max_{i,j} \Vert \widehat{\sigma}_{ij} - \sigma_{ij} \Vert_{2,\Pb} \right\} \right).
\end{align*}

The case for $\widehat{m}^a$ is straightforward and omitted here. Now we have the approximation-by-averages representation
\begin{align*}
\sqrt{n}
\begin{Bmatrix}
\begin{pmatrix}
\widehat{\psi}_{i} \\
\widehat{\psi}_{j} \\
\widehat{\psi}_{ij}
\end{pmatrix} -
\begin{pmatrix}
{\psi}_{i} \\
{\psi}_{j} \\
{\psi}_{ij}
\end{pmatrix}
\end{Bmatrix}
& = \sqrt{n}\left(\Pn - \Pb \right) 
\begin{pmatrix}
{\phi}_{i}(Z;\eta_i) \\
{\phi}_{j}(Z;\eta_j) \\
{\phi}_{ij}(Z;\eta_{ij})
\end{pmatrix}
+ o_\Pb(1) \\
&  \xrightarrow[]{d} N\left(0, \cov\begin{pmatrix}{\phi}_{i}(Z;\eta_i) \\ {\phi}_{j}(Z;\eta_j) \\ {\phi}_{ij}(Z;\eta_{ij}) \end{pmatrix}\right).
\end{align*}

Now for a vector $(y_1, y_2, y_3)^\top \in \R^3$, define a function $g:\R^3 \rightarrow \R$ such that $g(y_1, y_2, y_3)=y_3 - y_1y_2$. Also let $\psi \equiv (\psi_{i},\psi_{j},\psi_{ij})^\top$, $\widehat{\psi} \equiv (\widehat{\psi}_{i}, \widehat{\psi}_{j},\widehat{\psi}_{ij})^\top$, and $\phi \equiv (\phi_{ij},\phi_{i},\phi_{j})^\top$. Then by the delta method, it follows that
\begin{align*}
\sqrt{n}\left(g(\widehat{\psi}) - g(\psi) \right) & =  \sqrt{n}\left( \widehat{\Sigma}^a_{ij} - {\Sigma}^a_{ij} \right) \\
& \xrightarrow[]{d} [\nabla g(\psi)]^{\top}N\left(0, \cov\begin{pmatrix}{\phi}_{i}(Z;\eta_i) \\ {\phi}_{j}(Z;\eta_j) \\ {\phi}_{ij}(Z;\eta_{ij}) \end{pmatrix}\right).
\end{align*}
\end{proof}

\subsection{Proof of Theorem \ref{thm:consistency-cov-shrinkage}}

Recall that ${\Sigma}^*_S = \rho_1^*\mathbb{I} + \rho_2^*\widehat{\Sigma}^a$ where $(\rho_1^*, \rho_2^*)$ is the solution of the program \eqref{eqn:opt-shrinkage}. Let us define
\begin{align} \label{eqn:cov-shrinkage-true}
    \tilde{\Sigma}^*_S= \varpi^*\nu\mathbb{I} + (1-\varpi^*)\widehat{\Sigma}^a \quad \text{with} \quad \varpi^*\coloneqq\frac{\beta^2}{\delta^2},
\end{align}
where $\nu = \langle \Sigma^a, \mathbb{I} \rangle$, $\beta^2 = \Pb \Vert\widehat{\Sigma}^a - \Sigma^a \Vert_F^2$ and $\delta^2 = \Pb \Vert\widehat{\Sigma}^a - \nu \mathbb{I} \Vert_F^2$. The next lemma shows that $\tilde{\Sigma}^*_S$ converges in probability to ${\Sigma}^*_S$ under very weak conditions. 

\begin{lemma} \label{lem:shrinkage-lemma-1}
Suppose that $\widehat{\pi}, \widehat{\mu}_j, \widehat{\sigma}_{ij}$ are consistent. Then
\[
\Vert \tilde{\Sigma}^*_S - {\Sigma}^*_S \Vert_F = O_\Pb\left(\Vert \widehat{\pi} - \pi \Vert_{2,\Pb} \Sigma_{i,j=1}^k \left\{ \Vert  \Vert \widehat{\mu}_j - \mu_j \Vert_{2,\Pb} + \Vert \widehat{\sigma}_{ij} - \sigma_{ij} \Vert_{2,\Pb} \right\}\right),
\]
and thus $\tilde{\Sigma}^*_S \xrightarrow[]{p} {\Sigma}^*_S$ in the Frobenius norm.
\end{lemma}
\begin{proof}
Recall that we are interested in the following optimization program 
\begin{equation}
\label{eqn:opt-shrinkage-re}
\begin{aligned}
    & \text{minimize} \quad  \Pb\Vert \rho\nu\mathbb{I} + (1-\rho)\widehat{\Sigma}^a - \Sigma^a \Vert_F^2 \\
    & \text{over} \quad \rho, \nu \in \mathbb{R}.
\end{aligned}    
\end{equation}

Note that
\begin{align*}
    &\Pb\left\Vert \rho\nu\mathbb{I} + (1-\rho)\widehat{\Sigma}^a- \Sigma^a \right\Vert_F^2 \\
    &= \Pb\left\Vert \rho\left(\nu\mathbb{I} - \Sigma^a \right) + (1-\rho)\left(\widehat{\Sigma}^a- \Sigma^a\right) \right\Vert_F^2\\
    & = \rho^2 \left\Vert \nu\mathbb{I} - \Sigma^a \right\Vert_F^2 + (1-\rho)^2\Pb\left\Vert \widehat{\Sigma}^a- \Sigma^a \right\Vert_F^2 + 2\rho(1-\rho) \left\langle \left(\nu\mathbb{I} - \Sigma^a \right), \Pb(\widehat{\Sigma}^a- \Sigma^a) \right\rangle.
\end{align*}

Now, for a $k \times k$ real-valued matrix $\Omega$, define a function $f$ indexed by $\Omega$ as
\begin{align*}
    f(\rho,\nu;\Omega)= \rho^2 \left\Vert \nu\mathbb{I} - \Sigma^a \right\Vert_F^2 + (1-\rho)^2\Pb\left\Vert \widehat{\Sigma}^a- \Sigma^a \right\Vert_F^2 + 2\rho(1-\rho) \left\langle \left(\nu\mathbb{I} - \Sigma^a \right), \Omega \right\rangle.
\end{align*}

So, if we let $\widehat{\Omega}$ denote a matrix whose $(i,j)$-component is given by $\Pb(\widehat{\Sigma}^a_{ij}- \Sigma^a_{ij})$ ($1 \leq i,j \leq k$) and $\bm{0}_{k \times k}$ denote a matrix of $k^2$ zeros, then we may write
\begin{align*}
    &f\left(\rho,\nu;\widehat{\Omega}\right) = \rho^2 \left\Vert \nu\mathbb{I} - \Sigma^a \right\Vert_F^2 + (1-\rho)^2\Pb\left\Vert \widehat{\Sigma}^a- \Sigma^a \right\Vert_F^2 + 2\rho(1-\rho) \left\langle \left(\nu\mathbb{I} - \Sigma^a \right), \widehat{\Omega} \right\rangle, \\
    & f\left(\rho,\nu;\bm{0}_{k \times k}\right) = \rho^2 \left\Vert \nu\mathbb{I} - \Sigma^a \right\Vert_F^2 + (1-\rho)^2\Pb\left\Vert \widehat{\Sigma}^a- \Sigma^a \right\Vert_F^2.
\end{align*}

Now consider an unconstrained parametric program
\begin{equation}
\label{eqn:opt-shrinkage-param}
\begin{aligned}
    & \text{minimize} \quad  f\left(\rho,\nu; {\Omega} \right) \\
    & \text{over} \quad \rho, \nu \in \mathbb{R}
\end{aligned}    \tag{$\mathsf{P}(\Omega)$}  
\end{equation}
with ${\Omega}$ as the parameter. Since $f \in C^2$ with respect to $(\rho,\nu)$ and its Hessian is positive definite (note that we tacitly assumed $\widehat{\Sigma}^a \neq \Sigma^a$, otherwise the result is trivial), a local minimizer of the program \ref{eqn:opt-shrinkage-param} is Lipschitz stable. Let $\widehat{y}=(\widehat{\rho}, \widehat{\nu})$, $y_0=(\rho_0, \nu_0)$ be the solutions of $\mathsf{P}(\widehat{\Omega})$, $\mathsf{P}(\bm{0}_{k \times k})$, respectively. Then by Lemma A.1 in \citet{kim2025semiparametric}, it follows that
\begin{equation} \label{eqn:lem10-1}
    \begin{aligned}
    \Vert\widehat{y} - y_0 \Vert_2 &= O_{\Pb}\left( \left\Vert \widehat{\Omega} \right\Vert_2 \right) \\
    &= O_{\Pb}\left( \left\Vert \Pb\left\{\widehat{\Sigma}^a - \Sigma^a \right\} \right\Vert_F \right) \\
    &= O_{\Pb}\left( \sum_{i,j} \left\vert \Pb\left\{\widehat{\Sigma}^a_{ij} - \Sigma^a_{ij} \right\} \right\vert \right) \\
    & =O_\Pb\left( \sum_{i}\Vert \widehat{\pi} - \pi \Vert_{2,\Pb}\Vert \widehat{\mu}_i - \mu_i \Vert_{2,\Pb} + \sum_{j}\Vert \widehat{\pi} - \pi \Vert_{2,\Pb}\Vert \widehat{\mu}_j - \mu_j \Vert_{2,\Pb} + \sum_{ij} \Vert \widehat{\pi} - \pi \Vert_{2,\Pb}  \Vert \widehat{\sigma}_{ij} - \sigma_{ij} \Vert_{2,\Pb} \right)\\
    & =O_\Pb\left(\Vert \widehat{\pi} - \pi \Vert_{2,\Pb} \Sigma_{ij}\left\{ \Vert  \Vert \widehat{\mu}_j - \mu_j \Vert_{2,\Pb} + \Vert \widehat{\sigma}_{ij} - \sigma_{ij} \Vert_{2,\Pb} \right\}\right),
    \end{aligned}
\end{equation}
where the fourth line follows by rearranging the second-order remainder terms of the estimators $\psi_{ij},\psi_{i},\psi_{j}$, $1 \leq i,j \leq k$ defined in the appendix \ref{proof-prop-double-robust}.

The program $\mathsf{P}(\widehat{\Omega})$ is equivalent to \eqref{eqn:opt-shrinkage-re}. Moreover, using the same logic used in Theorem 2.1 of \citet{ledoit2004well}, it can be deduced that the solution of the program $\mathsf{P}(\bm{0}_{k \times k})$ is given by $(\varpi^*,\nu)$ defined in \eqref{eqn:cov-shrinkage-true}: i.e., $y_0 = (\varpi^*,\nu)$. Hence, by \eqref{eqn:lem10-1} and the given consistency conditions, we finally obtain the desired result:
\begin{align*}
\left\Vert \tilde{\Sigma}^*_S - {\Sigma}^*_S \right\Vert_F &= O_\Pb\left(\Vert \widehat{\pi} - \pi \Vert_{2,\Pb} \Sigma_{ij}\left\{ \Vert  \Vert \widehat{\mu}_j - \mu_j \Vert_{2,\Pb} + \Vert \widehat{\sigma}_{ij} - \sigma_{ij} \Vert_{2,\Pb} \right\}\right)\\
&= o_\Pb(1).
\end{align*}

\end{proof}

Next, we show that our proposed estimator \eqref{eqn:cov-shrinkage-est} converges in probability to $\tilde{\Sigma}^*_S$, which concludes the first part of our proof of Theorem \ref{thm:consistency-cov-shrinkage}.

\begin{lemma}\label{lem:shrinkage-lemma-2}
Let $\Vert \widehat{\pi} - \pi \Vert_{2,\Pb} \Sigma_{ij}\left\{ \Vert  \Vert \widehat{\mu}_j - \mu_j \Vert_{2,\Pb} + \Vert \widehat{\sigma}_{ij} - \sigma_{ij} \Vert_{2,\Pb} \right\} = O_\Pb(r(n))$. Then,
\[
\Vert \widehat{\Sigma}^*_S - \tilde{\Sigma}^*_S \Vert_F = O_\Pb\left(n^{-1/2} \vee r(n)\right).
\]
\end{lemma}

\begin{proof}
It suffices to show that $\widehat{\nu}, \widehat{\delta}, \widehat{\beta}$ are consistent at the specified rate. Let $\widehat{S}^a$ denote the (virtual) sample covariance matrix that can be computed from $Y^a_1, ..., Y^a_n$ (note that $\widehat{S}^a$ can never be computed in reality).

\textbf{i)} $\boldmath{\widehat{\nu} - \nu = O_\Pb\left(n^{-1/2} \vee r(n)\right)}$. 
$\forall 1 \leq i,j \leq k$, we have $\widehat{\Sigma}^a_{ij} - {\Sigma}^a_{ij} = O_\Pb(r(n))$ by Lemma \ref{lem:double-robust}. It also follows $\widehat{S}^a_{ij} - {\Sigma}^a_{ij} = O_\Pb(n^{-1/2})$ by the central limit theorem. Hence by the continuous mapping theorem,
\[
\Vert \widehat{\Sigma}^a - \widehat{S}^a \Vert_F = O_\Pb\left(n^{-1/2} \vee r(n)\right).
\]

Now we have
\begin{align*}
    \left\vert \widehat{\nu} - \nu \right\vert&= \left\vert \langle \widehat{\Sigma}^a, \mathbb{I} \rangle - \langle {\Sigma}^a, \mathbb{I} \rangle \right\vert\\
    &= \left\vert \left\langle \left(\widehat{\Sigma}^a - \widehat{S}^a + \widehat{S}^a \right), \mathbb{I} \right\rangle - \left\langle {\Sigma}^a, \mathbb{I} \right\rangle \right\vert\\
    & \leq \left\vert \left\langle \widehat{\Sigma}^a - \widehat{S}^a , \mathbb{I} \right\rangle \right\vert +\left\vert \left\langle \widehat{S}^a, \mathbb{I} \right\rangle - \left\langle {\Sigma}^a, \mathbb{I} \right\rangle \right\vert\\
    &= O_\Pb\left(n^{-1/2} \vee r(n)\right) + O_\Pb(n^{-1/2}),
\end{align*}
which yields the desired conclusion.

\textbf{ii)} $\boldmath{\widehat{\delta} - \delta=O_\Pb\left(n^{-1/2} \vee r(n)\right)}$. 
First note that
\begin{align*}
     & \widehat{\delta} -  \delta  \\
     &=  \Vert \widehat{\Sigma}^a - \widehat{\nu}\mathbb{I} \Vert_F - \left(\Pb \Vert\widehat{\Sigma}^a - \nu \mathbb{I} \Vert_F^2 \right)^{1/2} \\
    &= \left\Vert \widehat{\Sigma}^a - \widehat{S}^a  + \widehat{S}^a  - \langle \widehat{\Sigma}^a - \widehat{S}^a  + \widehat{S}^a , \mathbb{I} \rangle \mathbb{I} \right\Vert_F - \left(\Pb \Vert\widehat{\Sigma}^a - \widehat{S}^a + \widehat{S}^a - \nu \mathbb{I} \Vert_F^2 \right)^{1/2} \\
    & \leq \left\Vert \widehat{\Sigma}^a - \widehat{S}^a \right\Vert_F + \left\Vert \widehat{S}^a  - \langle \widehat{S}^a , \mathbb{I} \rangle \mathbb{I} \right\Vert_F + \left\Vert \langle \widehat{\Sigma}^a - \widehat{S}^a  , \mathbb{I} \rangle \mathbb{I} \right\Vert_F - \left(\Pb \Vert\widehat{S}^a - \nu \mathbb{I} \Vert_F^2 \right)^{1/2} + \left(\Pb \Vert\widehat{\Sigma}^a - \widehat{S}^a \Vert_F^2 \right)^{1/2},
\end{align*}
where the second last inequality follows by Jensen's Inequality and the last by the triangle inequality and the fact that $\sqrt{\Vert A + B \Vert_F^2} \geq \sqrt{\Vert A \Vert_F^2} - \sqrt{\Vert B \Vert_F^2}$, $\forall A,B \in \mathbb{R}^{k\times k}$.

From part i), we know $\left\Vert \widehat{\Sigma}^a - \widehat{S}^a \right\Vert_F =O_\Pb\left(n^{-1/2} \vee r(n)\right)$. Since $\left\Vert \langle \widehat{\Sigma}^a - \widehat{S}^a  , \mathbb{I} \rangle \mathbb{I} \right\Vert_F \leq \left\Vert \widehat{\Sigma}^a - \widehat{S}^a \right\Vert_F$, it follows $\left\Vert \langle \widehat{\Sigma}^a - \widehat{S}^a  , \mathbb{I} \rangle \mathbb{I} \right\Vert_F =O_\Pb\left(n^{-1/2} \vee r(n)\right)$. Thus the first and third terms in the last display converge at the desired rate.

For the fifth term, by the triangle inequality
\[
    \left(\Pb \Vert\widehat{\Sigma}^a - \widehat{S}^a \Vert_F^2 \right)^{1/2} \equiv \Vert\widehat{\Sigma}^a - \widehat{S}^a \Vert_{F,\Pb} \leq \Vert\widehat{\Sigma}^a - \Sigma^a \Vert_{F,\Pb} + \Vert \widehat{S}^a - \Sigma^a \Vert_{F,\Pb},
\]
where we view $\Vert \cdot \Vert_{F,\Pb}$ as an element-wise $L_2(\Pb)$-norm for matrix. By Proposition \ref{prop:bound-L2P-norm-Sigma}, we have
\begin{align*}
    \Vert\widehat{\Sigma}^a - \Sigma^a \Vert_{F,\Pb} &\leq \sum_{i,j} \frac{\Vert \widehat{\Sigma}^a_{ij} \Vert_{2,\Pb}}{\sqrt{n}}  + \sum_{i,j} \left\vert \Pb\left\{\widehat{\Sigma}^a_{ij} - \Sigma^a_{ij} \right\} \right\vert \\
    &= O_\Pb\left( \frac{1}{\sqrt{n}} \right) + o_\Pb\left( r(n) \right).
\end{align*}

Further, by Theorem 3.1 in \citet{ledoit2004well} it follows that  
\[
\Vert \widehat{S}^a - \Sigma^a \Vert_{F,\Pb} = O\left( \frac{1}{\sqrt{n}} \right).
\]

Therefore, we get $\Vert\widehat{\Sigma}^a - \widehat{S}^a \Vert_{F,\Pb} = O_\Pb\left(n^{-1/2} \vee r(n)\right)$. Bringing these results together, we have
\begin{align*}
    & \left\Vert \widehat{\Sigma}^a - \widehat{S}^a \right\Vert_F + \left\Vert \langle \widehat{\Sigma}^a - \widehat{S}^a  , \mathbb{I} \rangle \mathbb{I} \right\Vert_F + \left(\Pb \Vert\widehat{\Sigma}^a - \widehat{S}^a \Vert_F^2 \right)^{1/2} + \left\Vert \widehat{S}^a  - \langle \widehat{S}^a , \mathbb{I} \rangle \mathbb{I} \right\Vert_F - \left(\Pb \Vert\widehat{S}^a - \nu \mathbb{I} \Vert_F^2 \right)^{1/2} \\
    &= O_\Pb\left(n^{-1/2} \vee r(n)\right) + \left\Vert \widehat{S}^a  - \langle \widehat{S}^a , \mathbb{I} \rangle \mathbb{I} \right\Vert_F - \left(\Pb \Vert\widehat{S}^a - \nu \mathbb{I} \Vert_F^2 \right)^{1/2}. 
\end{align*}

However, since we consider the case of fixed $p$, Lemma 3.3 of \citet{ledoit2004well} implies
\[
\left\Vert \widehat{S}^a  - \langle \widehat{S}^a , \mathbb{I} \rangle \mathbb{I} \right\Vert_F - \left(\Pb \Vert\widehat{S}^a - \nu \mathbb{I} \Vert_F^2 \right)^{1/2} = O\left( \frac{1}{\sqrt{n}} \right),
\]
which finally leads to 
\[
\widehat{\delta} -  \delta = O_\Pb\left(n^{-1/2} \vee r(n)\right).
\]

Similarly, one can also show that
\begin{align*}
\delta - \widehat{\delta} & \leq  \left(\Pb \Vert\widehat{S}^a - \nu \mathbb{I} \Vert_F^2 \right)^{1/2} - \left\Vert \widehat{S}^a  - \langle \widehat{S}^a , \mathbb{I} \rangle \mathbb{I} \right\Vert_F  + O_\Pb\left(n^{-1/2} \vee r(n)\right) \\
&=  O_\Pb\left(n^{-1/2} \vee r(n)\right).
\end{align*}

Hence, we obtain $\widehat{\delta} -  \delta = O_\Pb\left(n^{-1/2} \vee r(n)\right)$. 

\textbf{iii)} $\boldmath{\widehat{\beta} - \beta = O_\Pb\left(n^{-1/2} \vee r(n)\right)}$. 
This in fact follows because with fixed $k$, each of $\widehat{\beta}$, $\beta$ vanishes quickly to zero. To show this, first let $\widehat{S}^a = \sum_{t=1}^n \widehat{S}_t^a$ and consider the following quantity
\begin{align*}
     \sqrt{\Pb\Vert\widehat{\Sigma}^a - \widehat{S}^a + \widehat{S}^a - \Sigma^a  \Vert_F^2}.
\end{align*}

By the Cauchy-Schwarz inequality we have
\begin{align*}
    \Pb\Vert\widehat{\Sigma}^a - \widehat{S}^a + \widehat{S}^a - \Sigma^a  \Vert_F^2 
    &= \Pb\Vert\widehat{\Sigma}^a - \widehat{S}^a \Vert_F^2 + 2\Pb \left\{\sum_{ij} (\widehat{\Sigma}^a_{ij} - \widehat{S}^a_{ij})(\widehat{S}^a_{ij}-\Sigma^a_{ij}) \right\} + \Pb\Vert \widehat{S}^a-\Sigma^a \Vert_F^2\\
    &\leq \Pb\Vert\widehat{\Sigma}^a - \widehat{S}^a \Vert_F^2 + 2\Pb \left\{\Vert (\widehat{\Sigma}^a - \widehat{S}^a)\Vert_F\Vert (\widehat{S}^a - {\Sigma}^a )\Vert_F\right\}  + \Pb\Vert \widehat{S}^a-\Sigma^a \Vert_F^2\\
    &\leq \Pb\Vert\widehat{\Sigma}^a - \widehat{S}^a \Vert_F^2 + 2\sqrt{\Pb\Vert (\widehat{\Sigma}^a - \widehat{S}^a)\Vert_F^2}\sqrt{\Pb\Vert (\widehat{S}^a-\Sigma^a)\Vert_F^2}  + \Pb\Vert \widehat{S}^a-\Sigma^a \Vert_F^2.
\end{align*}

In part ii), we showed $\sqrt{\Pb\Vert (\widehat{\Sigma}^a - \widehat{S}^a)\Vert_F^2} =O_\Pb\left(n^{-1/2} \vee r(n)\right)$ and $\sqrt{\Pb\Vert (\widehat{S}^a - {\Sigma}^a )\Vert_F^2} =O\left(n^{-1/2}\right)$. Hence it follows
\begin{align*}
    \sqrt{\Pb\Vert (\widehat{\Sigma}^a - \widehat{S}^a)\Vert_F^2}\sqrt{\Pb\Vert (\widehat{S}^a-\Sigma^a)\Vert_F^2}
    &= O_\Pb\left(n^{-1/2} \vee r(n)\right)O_\Pb\left( n^{-1/2} \right),
\end{align*}
and consequently, we have
\begin{align*}
    \sqrt{\Pb\Vert\widehat{\Sigma}^a - \widehat{S}^a + \widehat{S}^a - \Sigma^a  \Vert_F^2}
    = O_\Pb\left(n^{-1/2} \vee r(n)\right).
\end{align*}

Next, we consider
\begin{align*}
    \sqrt{\frac{1}{n^2}\sum_{t=1}^n \Vert \widetilde{\Sigma}_t -\widehat{S}_t^a + \widehat{S}_t^a - \widehat{S}^a + \widehat{S}^a - \widehat{\Sigma}^a \Vert_F^2}.
\end{align*}

We shall first show that
\[
\frac{1}{n}\sum_{t=1}^n \Vert \widetilde{\Sigma}_t -\widehat{S}_t^a \Vert_F^2 = O_\Pb(1).
\]

To this end, we note that
\begin{align*}
    \frac{1}{n}\sum_{t=1}^n \Vert \widetilde{\Sigma}_t - \Sigma^a \Vert_F^2 &= \sum_{ij}\left\{\frac{1}{n}\sum_{t=1}^n \left( \widehat{\phi}^a_{ij}(Z_t) -  \Pn \widehat{\phi}^a_{i} \Pn \widehat{\phi}^a_{j} - \Sigma^a_{ij} \right)^2 \right\}\\
    &= \sum_{ij}\left\{\frac{1}{n}\sum_{t=1}^n \left( \widehat{\phi}^a_{ij}(Z_t) -  \Pn \widehat{\phi}^a_{i} \Pn \widehat{\phi}^a_{j} - \Pn\widehat{\phi}^a_{ij} + \Pn \widehat{\phi}^a_{i} \Pn \widehat{\phi}^a_{j} +\Pn\widehat{\phi}^a_{ij} - \Pn \widehat{\phi}^a_{i} \Pn \widehat{\phi}^a_{j}  - \Sigma^a_{ij} \right)^2 \right\}\\
    &= \sum_{ij}\left\{\frac{1}{n}\sum_{t=1}^n \left( \widehat{\phi}^a_{ij}(Z_t)  - \Pn\widehat{\phi}^a_{ij}  +\Pn\widehat{\phi}^a_{ij} - \Pn \widehat{\phi}^a_{i} \Pn \widehat{\phi}^a_{j}  - \Sigma^a_{ij} \right)^2 \right\}\\
    &\leq \sum_{ij}\Bigg\{\frac{1}{n}\sum_{t=1}^n \left( \widehat{\phi}^a_{ij}(Z_t)  - \Pn\widehat{\phi}^a_{ij} \right)^2 + \left(\Pn\widehat{\phi}^a_{ij} - \Pn \widehat{\phi}^a_{i} \Pn \widehat{\phi}^a_{j}  - \Sigma^a_{ij} \right)^2  \\
    & \quad + \left\vert \Pn\widehat{\phi}^a_{ij} - \Pn \widehat{\phi}^a_{i} \Pn \widehat{\phi}^a_{j}  - \Sigma^a_{ij} \right\vert \sqrt{\frac{1}{n}\sum_{t=1}^n \left( \widehat{\phi}^a_{ij}(Z_t)  - \Pn\widehat{\phi}^a_{ij} \right)^2} \Bigg\},
\end{align*}
and that $\forall 1 \leq i,j \leq k$,
\begin{align*}
    & \frac{1}{n}\sum_{t=1}^n \left( \widehat{\phi}^a_{ij}(Z_t)  - \Pn\widehat{\phi}^a_{ij} \right)^2 \xrightarrow[]{p} \var\left(\widehat{\phi}^a_{ij}(Z) \mid Z_{n+1},...,Z_{2n} \right)\\
    & \Pn\widehat{\phi}^a_{ij} - \Pn \widehat{\phi}^a_{i} \Pn \widehat{\phi}^a_{j}  - \Sigma^a_{ij} = O_\Pb(1).
\end{align*}

Then from the given boundedness conditions, $\var\left(\widehat{\phi}^a_{ij}(Z) \mid Z_{n+1},...,Z_{2n} \right)$ is bounded in probability, and thus so is the RHS of the last inequality. Hence $\frac{1}{n}\sum_{t=1}^n \Vert \widetilde{\Sigma}_t -\widehat{S}_t^a \Vert_F^2 = O_\Pb(1)$.

Next, by the unbiasedness of the sample covariance estimator we get
\begin{align*}
    \frac{1}{n}\sum_{t=1}^n \Vert \widehat{S}_t^a  - \Sigma^a \Vert_F^2 = \sum_{ij}\left\{\frac{1}{n}\sum_{t=1}^n \left( \widehat{S}^a_{t,ij} - \Sigma^a_{ij} \right)^2 \right\}
    \xrightarrow[]{p} \sum_{ij} \var\left(\widehat{S}^a_{ij}\right) < \infty
\end{align*}
as $Y_i$'s have finite fourth moments. Thus $\frac{1}{n}\sum_{t=1}^n \Vert \widehat{S}_t^a  - \Sigma^a \Vert_F^2 = O_\Pb(1)$.

Using these facts, and again by the Cauchy Schwarz inequality, we obtain
\[
\frac{1}{n}\sum_{t=1}^n \Vert \widetilde{\Sigma}_t -\widehat{S}_t^a \Vert_F^2 = O_\Pb(1)
\]
as desired.

Moreover, the terms involving $\Vert \widehat{\Sigma}^a - \widehat{S}^a \Vert_F$ converge at fast rates since
$\Vert \widehat{\Sigma}^a - \widehat{S}^a \Vert_F = O_\Pb\left(n^{-1/2} \vee r(n)\right)$ as shown in part i). Therefore, we have
\begin{align*}
    \frac{1}{n^2}\sum_{i=1}^n \Vert \widetilde{\Sigma}_i -\widehat{S}_i^a + \widehat{S}_i^a - \widehat{S}^a + \widehat{S}^a - \widehat{\Sigma}^a \Vert_F^2 &=
    O\left(n^{-1}\right) +\frac{1}{n^2}\sum_{i=1}^n \Vert \widehat{S}_i^a - \widehat{S}^a \Vert_F^2,
\end{align*}
which follows by simple rearrangement and the Cauchy Schwarz inequality.

On the other hand, Lemma 3.4 of \citet{ledoit2004well} indicates that
\[
\frac{1}{n^2}\sum_{i=1}^n \Vert \widehat{S}_i^a - \widehat{S}^a \Vert_F^2 - \Pb\Vert \widehat{S}^a-\Sigma^a \Vert_F^2 = O_\Pb\left(n^{-1} \right).
\]

Now we can bring all the results together, to get to the conclusion
\begin{align*}
    \sqrt{\frac{1}{n^2}\sum_{t=1}^n \Vert \widetilde{\Sigma}_t -\widehat{S}_t^a + \widehat{S}_t^a - \widehat{S}^a + \widehat{S}^a - \widehat{\Sigma}^a \Vert_F^2}
    &= O_\Pb\left(n^{-1/2} \vee r(n)\right),
\end{align*}
which completes the proof.
\end{proof}

The followings are the auxiliary technical results used for the proof of Lemma \ref{lem:shrinkage-lemma-2}.


\begin{lemma} \label{lem:bound-L2P-norm}
Let $\Pn$ denote the empirical measure over an iid sample $(Z_{1},\ldots,Z_{n})$.
Also we let $f$ and $\hat{f}$ be any function and its estimator constructed in a separate, independent sample $(Z_{n+1},\ldots,Z_{2n})$, respectively. Then we have
\[
\Vert \Pn\hat{f} - \Pb f \Vert_{2,\Pb} \leq \frac{\left\Vert \hat{f}\right\Vert_{2,\Pb}}{\sqrt{n}} + \Pb\left(\hat{f} - f\right).
\]
\end{lemma}
\begin{proof}
\begin{align*}
    \Vert \Pn\hat{f} - \Pb f \Vert_{2,\Pb} &= \Vert (\Pn - \Pb)\hat{f} - \Pb(\hat{f} - f) \Vert_{2,\Pb} \\
    & \leq  \Vert (\Pn - \Pb)\hat{f} \Vert_{2,\Pb} - \Vert \Pb(\hat{f} - f) \Vert_{2,\Pb}\\
    & \leq \sqrt{\frac{\var\left[\hat{f}\mid Z_{n+1},\ldots,Z_{2n} \right]}{n}} + \Vert \Pb(\hat{f} - f) \Vert_{2,\Pb} \\
    & \leq \frac{\left\Vert \hat{f}\right\Vert_{2,\Pb}}{\sqrt{n}} + \Pb\left(\hat{f} - f\right),
\end{align*}
where the third line follows by Lemma C.3 in \citet{kim2018causal}. 
\end{proof}

\begin{proposition}\label{prop:bound-L2P-norm-Sigma}
Suppose that $\forall 1 \leq i,j \leq k$, $\Vert \widehat{\sigma}_{ij} \Vert_{2,\Pb} = O_\Pb(1)$ and $\Vert \widehat{\mu}_{i} \Vert_{2,\Pb} = O_\Pb(1)$. Then, we have
\[
\Vert \widehat{\Sigma}^a_{ij} - {\Sigma}^a_{ij} \Vert_{2,\Pb} = O_\Pb\left(\frac{1}{\sqrt{n}}\right) + o_\Pb\left(r(n)\right),
\]
where $r(n)$ is defined in Lemma \ref{lem:shrinkage-lemma-2}.
\end{proposition} 
\begin{proof}
Recall that for our counterfactual covariance estimator $\widehat{\Sigma}^a_{ij}$ defined in \eqref{eqn:cov-estimator}, we can write
\[
\widehat{\Sigma}^a_{ij} -  {\Sigma}^a_{ij} = \Pn \widehat{\phi}^a_{ij} -  \Pn \widehat{\phi}^a_{i} \Pn \widehat{\phi}^a_{j}  - \Pb {\phi}^a_{ij} - \Pb {\phi}^a_{i} \Pb {\phi}^a_{j}.
\]
Hence, by the result of Lemma \ref{lem:bound-L2P-norm} we obtain
\begin{align*}
    \Vert \widehat{\Sigma}^a_{ij} - {\Sigma}^a_{ij} \Vert_{2,\Pb} &= \left\Vert \Pn \widehat{\phi}^a_{ij} -  \Pn \widehat{\phi}^a_{i} \Pn \widehat{\phi}^a_{j}  - \Pb {\phi}^a_{ij} - \Pb {\phi}^a_{i} \Pb {\phi}^a_{j} \right\Vert_{2,\Pb}\\
    &\leq  \Vert \Pn \widehat{\phi}^a_{ij}  - \Pb {\phi}^a_{ij}  \Vert_{2,\Pb} + \Vert \Pn \widehat{\phi}^a_{i} \left\{\Pn \widehat{\phi}^a_{j}  - \Pb {\phi}^a_{j}\right\}  \Vert_{2,\Pb} + \Vert \Pn \widehat{\phi}^a_{j} \left\{\Pn \widehat{\phi}^a_{i}  - \Pb {\phi}^a_{i} \right\}  \Vert_{2,\Pb}\\
    &\leq \Vert \Pn \widehat{\phi}^a_{ij}  - \Pb {\phi}^a_{ij}  \Vert_{2,\Pb} + \Vert \Pn \widehat{\phi}^a_{i} \Vert_{2,\Pb} \left\Vert\Pn \widehat{\phi}^a_{j}  - \Pb {\phi}^a_{j}\right\Vert_{2,\Pb} + \Vert \Pn \widehat{\phi}^a_{j} \Vert_{2,\Pb} \left\Vert \Pn \widehat{\phi}^a_{i}  - \Pb {\phi}^a_{i} \right\Vert_{2,\Pb} \\
    & \lesssim \frac{\left\Vert \widehat{\phi}^a_{ij}\right\Vert_{2,\Pb}+\left\Vert \widehat{\phi}^a_{i}\right\Vert_{2,\Pb}\left\Vert \widehat{\phi}^a_{i}\right\Vert_{2,\Pb}}{\sqrt{n}} + \left(1 + \left\Vert \widehat{\phi}^a_{i}\right\Vert_{2,\Pb} + \left\Vert \widehat{\phi}^a_{j}\right\Vert_{2,\Pb}\right)o_\Pb\left(r(n)\right) \\
    &= O_\Pb\left(\frac{1}{\sqrt{n}}\right) + O_\Pb\left(1\right)o_\Pb\left(r(n)\right)
\end{align*}
, which gives the result.
\end{proof}

The second part of the proof of Theorem \ref{thm:consistency-cov-shrinkage} immediately follows by Theorem \ref{cor:rates-opt-sol-MV} and Theorem 2 in \citet{kim2025semiparametric}:
\begin{align*}
\Vert \widehat{w}_S - w^* \Vert_2 & \leq \Vert \widehat{w} - w^* \Vert_2 + \Vert \widehat{w}_S - \widehat{w} \Vert_2 \\
& = O_\Pb\left(r_n \vee n^{-1/2}\right) + O_\Pb\left( \Vert \widehat{\Sigma}^*_S - \widehat{\Sigma}^a \Vert_F \right)\\
&= O_\Pb\left(r_n \vee n^{-1/2}\right).
\end{align*}

\subsection{Proof of Theorem \ref{thm:pd-correction}}

For arbitrary $k \times k$ matrix $\Sigma$ and $k \times 1$ vector $m$, define a parametric program $\mathsf{P}(\Sigma,m)$
\begin{equation}
\label{eqn:couterfactual-MV-param}
\begin{aligned}
    & \underset{w \in \mathcal{W}}{\text{minimize}} \quad {1}/{2}w^\top\Sigma w  - \tau w^\top m \\
    & \text{subject to} \quad w \in \mathcal{S}(\Sigma,m),
\end{aligned}     \tag{$\mathsf{P}(\Sigma,m)$}  
\end{equation}
by viewing $\Sigma$ and $m$ together as parameters. Then $\mathsf{P}(\Sigma^a, m^a)$,  $\mathsf{P}(\widehat{\Sigma}^*_{cor}, \widehat{m}^a)$ are our true and approximating programs, respectively. 

\begin{proof}
By virtue of the quadratic growth condition, $\sol\left({\mathsf{P}(\Sigma^a, m^a)}\right)$ is a singleton \citep[][Theorem 2.5]{still2018lectures}, and $w^* = \sol\left({\mathsf{P}(\Sigma^a, m^a)}\right)$. By the Lipschitz stability result for smooth ($C^2$) parametric programming \citep[][Theorem 6.2]{still2018lectures}, for each $w^*$ there exist $\varepsilon, L > 0$ such that for all $\bar{\Sigma} \in \mathbb{B}_{\varepsilon}(\Sigma^a)$, $\bar{m} \in \mathbb{B}_{\varepsilon}(m^a)$ there exists at least one local minimizer $\bar{x}$ of $\mathsf{P}(\bar{\Sigma}, \bar{m})$ that satisfies
\[
\Vert\bar{x} - w^*  \Vert_2 \leq L\left\{ \Vert \bar{\Sigma} - \Sigma^a \Vert_2 + \Vert \bar{m} - m^a \Vert_2 \right\}.
\]

Now, it is straightforward to see that
\begin{align*}
   \Vert \widehat{w}_{cor} - {w^*} \Vert_2 & \leq   \Vert \widehat{w}_{cor} - {w^*} \Vert_2  \mathbbm{1}\left\{\widehat{\Sigma}^*_{cor} \in \mathbb{B}_{\varepsilon}(\Sigma^a), \widehat{m}^a \in \mathbb{B}_{\varepsilon}(m^a)\right\} \\
   & \quad +  \Vert \widehat{w}_{cor} - {w^*} \Vert_2 \left( \mathbbm{1}\left\{\widehat{\Sigma}^*_{cor} \notin \mathbb{B}_{\varepsilon}(\Sigma^a)\right\} + \mathbbm{1}\left\{\widehat{m}^a \notin \mathbb{B}_{\varepsilon}(m^a)\right\} \right).
\end{align*}

For the first term, it follows that
\begin{align*}
    \Vert \widehat{w}_{cor} - {w^*} \Vert_2 \mathbbm{1}\left\{\widehat{\Sigma}^*_{cor} \in \mathbb{B}_{\varepsilon}(\Sigma^a), \widehat{m}^a \in \mathbb{B}_{\varepsilon}(m^a)\right\} & \leq L \left( \Vert \widehat{\Sigma}^*_{cor} - \Sigma^a \Vert_2 + \Vert \widehat{m} - m^a \Vert_2 \right).
\end{align*}

Also $\forall \epsilon', \varepsilon >0$, we have that
\begin{align*}
   \Pb\left(\mathbbm{1}\left\{\widehat{\Sigma}^*_{cor} \notin \mathbb{B}_{\varepsilon}(\Sigma^a)\right\} > \epsilon' \right) &= \Pb\left( \widehat{\Sigma}^*_{cor} \notin \mathbb{B}_{\varepsilon}(\Sigma^a) \right)\\
   &= \Pb\left( \Vert \widehat{\Sigma}^*_{cor} - \Sigma^a \Vert_2 \geq \varepsilon \right) \\
   & \leq \Pb\left( \Vert \widehat{\Sigma}^*_{cor} - \widehat{\Sigma}^a \Vert_2 + \Vert \widehat{\Sigma}^a - \Sigma^a \Vert_2 \geq \varepsilon \right) \\
   & \rightarrow 0,
\end{align*}
which follows by the fact that $\Vert \widehat{\Sigma}^*_{cor} - \widehat{\Sigma}^a \Vert_2 + \Vert \widehat{\Sigma}^a - \Sigma^a \Vert_2 = o_\Pb(1)$ under the given conditions. Similarly, we obtain $\mathbbm{1}\left\{\widehat{m}^a \notin \mathbb{B}_{\varepsilon}(m^a)\right\} = o_\Pb(1)$ which immediately follows by that $\Vert \widehat{m} - m^a \Vert_2 = o_\Pb(1)$.

Putting the pieces together, we obtain 
\begin{align*}
   \Vert \widehat{w}_{cor} - {w^*} \Vert_2 &=   O\left(\Vert \widehat{\Sigma}^*_{cor} - \Sigma^a \Vert_2 + \Vert \widehat{m} - m^a \Vert_2 \right) + o_\Pb\left( \Vert \widehat{w}_{cor} - {w^*} \Vert_2 \right) 
\end{align*}
which yields
\begin{align*}
   \frac{\Vert \widehat{w}_{cor} - {w^*} \Vert_2}{\Vert \widehat{\Sigma}^*_{cor} - \Sigma^a \Vert_2+\Vert \widehat{m} - m^a \Vert_2} &=   O_\Pb\left( 1 \right).
\end{align*}
Hence, the result follows. 
\end{proof}

\subsection{Alternate proof of Theorem \ref{thm:pd-correction}}
\begin{proof}
Since $\Sigma^a$ is assumed to be PD, the quadratic growth condition (A2) holds at $w^*$, so we have
\begin{align*}
	\left\Vert \widehat{w}_{cor}- w^* \right\Vert_2 &=  O_\Pb\left(\Vert\widehat{\Sigma}^*_{cor} - \Sigma^a\Vert_2 \vee r_n\vee n^{-1/2} \right) \\
	\implies \left\Vert \widehat{w}_{cor}- w^* \right\Vert_2 &=  O_\Pb\left(\Vert\widehat{\Sigma}^*_{cor} - \widehat{\Sigma}^a\Vert_2 + \Vert\widehat{\Sigma}^a - \Sigma^a\Vert_2 \vee r_n\vee n^{-1/2} \right) \\
	&=  O_\Pb\left(\Vert\widehat{\Sigma}^*_{cor} - \widehat{\Sigma}^a\Vert_2\vee r_n\vee n^{-1/2} \right)
\end{align*}
where the first line follows from Theorem 3.1, and the last line follows from Lemma 4.1 which says that $\Vert\widehat{\Sigma}^a - \Sigma^a\Vert_2 = O_\Pb(r_n) + O_\Pb(n^{-1/2})$.

For the second statement of the theorem, suppose that $\widehat{\Sigma}^*_{cor} = \widehat{\Sigma}^a$ whenever $\widehat{\Sigma}^a$ is PD. Let $S_+^k$ denote the set of all PD $k\times k$ matrices, and for any $\epsilon > 0$ let $B_\epsilon(\Sigma^a) = \{M \in S^{k\times k} : \Vert M - \Sigma^a \Vert_F \leq \epsilon\}$ denote the $\epsilon$-ball in Frobenius norm of all symmetric matrices around $\Sigma^a$. Since $S_k^+$ is an open set, we can fix a $\delta > 0$ such that $B_\delta(\Sigma^a) \subset S_+^k$. We have
\begin{align*}
 \Pb(\widehat{\Sigma}^*_{cor} \neq \widehat{\Sigma}^a)
	&\leq \Pb(\widehat{\Sigma}^a \not\in B_\delta(\Sigma^a)) \\
	&= \Pb(\Vert\widehat{\Sigma}^a - \Sigma^a\Vert_F > \delta) \\
	&\rightarrow 0 \text{ as } n \rightarrow \infty
\end{align*}
where the last line again follows from Lemma 4.1. Next, note that for any sequence of positive numbers $a_n$,
\begin{align*}
	\Pb(a_n^{-1}\Vert\widehat{\Sigma}^*_{cor} - \widehat{\Sigma}^a\Vert_F > \epsilon) &= \Pb(a_n^{-1}\Vert\widehat{\Sigma}^*_{cor} - \widehat{\Sigma}^a\Vert > \epsilon \mid \widehat{\Sigma}^*_{cor} \neq \widehat{\Sigma}^a)\Pb(\widehat{\Sigma}^*_{cor} \neq \widehat{\Sigma}^a) \\
	&\leq \Pb(\widehat{\Sigma}^*_{cor} \neq \widehat{\Sigma}^a) \\
	&\rightarrow 0 \text{ as } n \rightarrow \infty
\end{align*}
so that $\Vert\widehat{\Sigma}^*_{cor} - \widehat{\Sigma}^a\Vert_F = o_\Pb(a_n)$. Letting $a_n = r_n \vee n^{-1/2}$, we have that
\begin{align*}
    \left\Vert \widehat{w}_{cor}- w^* \right\Vert_2 &=  O_\Pb\left(r_n\vee n^{-1/2} \right)
\end{align*}
as claimed.

\end{proof}